\documentclass{toc}

\tocdetails{%
  title = {Sublinear-Time Computation in the Presence of Online Erasures},
  number_of_pages = {41}, 
  number_of_bibitems = {69},
  number_of_figures = {1},
  conference_version = {ITCS22}, %
  author = {Iden Kalemaj, Sofya Raskhodnikova, and Nithin Varma},
  authorlist = {Iden Kalemaj, Sofya Raskhodnikova, Nithin Varma},
  acmclassification = {F.2},
  amsclassification = {68Q25},
  keywords = {randomized algorithms, property testing, Fourier analysis, linear functions, quadratic functions, Lipschitz and monotone functions, sorted sequences},
}   %

\tocdetails{%
volume = 19,      %
number = 1,       %
year = 2023,
received = {June 13, 2022},
revised = {June 29, 2023},
published = {August 23, 2023},
}   %

\usepackage{graphicx}   %

\usepackage{soul}
\usepackage{algorithm}
\usepackage[english]{babel}
\usepackage[noend]{algpseudocode}
\usepackage [autostyle, english = american]{csquotes}
\usepackage{ulem} 
\usepackage[shortlabels]{enumitem}
\usepackage[nameinlink,capitalize]{cleveref}
\usepackage{setspace}
\usepackage{caption}

\usepackage{aumacros}

\begin{document}

\begin{frontmatter}%

  \title{Sublinear-Time Computation in the Presence of Online Erasures\titlefootnote{A preliminary version of this paper appeared in
  the \href{https://doi.org/10.4230/LIPIcs.ITCS.2022.90}{Proceedings of the 13th ``Innovations in Theoretical Computer Science'' Conference (ITCS'22).}~\cite{KalemajRV22}.}}

\author[kalemaj]{Iden Kalemaj\thanks{Supported by NSF award CCF-1909612 and Boston University's Dean's Fellowship.}}
\author[raskhodnikova]{Sofya Raskhodnikova\thanks{Supported by NSF award CCF-1909612.}}
\author[varma]{Nithin Varma}

\begin{abstract}
  We initiate the study of sublinear-time algorithms that access their input via an online adversarial erasure oracle. After answering each input query,
	   such an oracle can erase $t$ input values. Our goal is to understand the complexity of basic computational tasks in extremely adversarial situations, where the algorithm's access to data is blocked during the execution of the algorithm in response to its actions. Specifically, we focus on property testing in the model with online erasures. We show that two fundamental properties of functions, linearity and quadraticity, can be tested for constant $t$ with asymptotically the same complexity as in the standard property testing model. For linearity testing, we prove tight bounds  in terms of $t$, showing that the query complexity is $\Theta(\log t).$ In contrast to linearity and quadraticity, some other properties, including sortedness and the Lipschitz property of sequences, cannot be tested at all, even for $t=1$.  Our investigation leads to a deeper understanding of the structure of violations of linearity and other widely studied properties. We also consider implications of our results for algorithms that are resilient to online adversarial corruptions instead of erasures.
\end{abstract}

%

\end{frontmatter}

\section{Introduction}

We initiate the study of sublinear-time algorithms that compute in the presence of an online adversary that blocks access to some data points in response to the algorithm's queries. A motivating scenario is when a user wishes to remove their data from a dataset due to privacy concerns, as enabled by right to be forgotten laws such as the EU General Data Protection Regulation \cite{Mantelero13}. 
The online aspect of our model suitably captures the case of individuals who are prompted to restrict access to their data after noticing an inquiry into their or others' data. The user could decide to remove their data after noticing, for example, that their phone number is available on websites that scrape personal contact information.
We choose to model such user actions as adversarial in order to perform worst-case analysis and refrain from making distributional and other assumptions on how the data access is affected. We give two other motivating scenarios that are naturally adversarial and justify a worst-case analysis. In one, an  algorithm is trying to detect some fraud (e.g., tax fraud) and the adversary wants to obstruct access to data in order to make it hard to uncover any evidence. In the other scenario, an algorithm's goal is to determine an optimal course of action (e.g., whether to invest in a stock or to buy an item), whereas the adversary leads the algorithm astray by adaptively blocking access to pertinent information.

In our model, after answering each query to the input object, the adversary can hide a small number of input values. Our goal is to understand the complexity of basic computational tasks in extremely adversarial situations, where the algorithm's access to data is blocked during the execution of the algorithm in response to its actions.
Specifically, we represent the input object as a function $f$ on an arbitrary finite domain\footnote{Input objects such as strings, sequences, images, matrices, and graphs can all be represented as functions. For example, an $n \times n$ real-valued matrix is equivalent to a function mapping $[n]^2$ to $\mathbb{R}$.}, which the algorithm can access by querying a point $x$ from the domain and receiving the answer $\cO(x)$ from %
an
oracle. At the beginning of computation, $\cO(x)=f(x)$ for all points $x$ in the domain of the function. We parameterize our model by a natural number $t$ that controls the number of function values the adversary %
can erase {\em after} the oracle %
answers each query\footnote{If the adversary were allowed to erase the query of the algorithm before answering it, the algorithm would only see erased values. We give several motivating scenarios for the adversarial behavior in our model. The first example is a situation where the adversary reacts by deleting additional data after some bank records are pulled by authorities as part of an investigation. In the GDPR example mentioned previously, we argued that individuals could be prompted to restrict access to their data only after noticing an inquiry into their or others' data. Finally, in a legal setting, if the adversary is served a subpoena, they are legally bound to answer the query, but could nonetheless destroy related evidence that is not included in the subpoena.}. 
Mathematically, we represent the oracle and the adversary as one entity. However, it might be helpful to think of the oracle as the data holder and of the adversary as the obstructionist.
A {\em $t$-online-erasure} oracle can replace values $\cO(x)$ on up to $t$ points $x$ with a special symbol $\perp$, thus erasing them.
The new values will be used by the oracle to answer future queries to the corresponding points. The locations of erasures are unknown to the algorithm.
The actions of the oracle can depend on the input, the queries made so far, and even on the publicly known code that the algorithm is running, but {\em not} on future coin tosses of the algorithm.

We focus on investigating property testing in the presence of online erasures. In the property testing model, introduced by \cite{RubinfeldS96,GGR98} with the goal of formally studying sublinear-time algorithms, a property is represented by a set  $\cP$ (of functions satisfying the desired property).  A function $f$ is $\eps$-far from $\cP$ if $f$ differs from each function $g\in\cP$ on at least an $\eps$ fraction of domain points. 
The goal is to distinguish, with constant probability, functions $f\in\cP$ from functions that are $\eps$-far from $\cP.$
We call an algorithm a {\em $t$-online-erasure-resilient $\eps$-tester} for property~$\cP$ if, given parameters $t\in\mathbb{N}$ and $\eps\in(0,1),$ and access to an input function $f$ via a $t$-online-erasure oracle, the algorithm accepts with probability at least 2/3 if $f\in\cP$ and rejects with probability at least 2/3 if $f$ is $\eps$-far from $\cP$.

We study the query complexity of online-erasure-resilient testing of several fundamental properties. 
We show that for linearity and quadraticity of functions $f:\{0,1\}^d\to\{0,1\}$, the query complexity of $t$-online-erasure-resilient testing for constant $t$ is asymptotically the same as in the standard model. For linearity, we also prove tight bounds in terms of $t$, showing that the query complexity is $\Theta(\log t)$.
A function $f(x)$ is {\em linear} if it can be represented as a sum of monomials of the form $x[i]$, where~$x=(x[1],\dots,x[d])$ is a vector of~$d$ bits; the function is {\em quadratic} if it can be represented as a sum of monomials of the form~$x[i]$ or~$x[i]x[j]$.

To understand
the difficulty of testing in the presence of online erasures, consider the case of linearity and $t=1.$ The celebrated tester for linearity in the standard property testing model was proposed by
Blum, Luby, and Rubinfeld~\cite{BlumLR93}. It looks for witnesses of non-linearity that consist of three points $x,y,$ and $x\oplus y$ satisfying $f(x)+f(y)\neq f(x\oplus y)$, where addition is mod 2, and $\oplus$ denotes bitwise XOR. 
Bellare et al.~\cite{BellareCHKS96} show that if $f$ is $\eps$-far from linear, then a triple $(x,y, x\oplus y)$ is a witness to non-linearity with probability at least $\eps$ when $x,y\in\{0,1\}^d$ are chosen uniformly and independently at random. 
In our model, after $x$ and $y$ are queried, the oracle can erase the value of $x\oplus y.$ 
To overcome this, our tester considers witnesses with more points, namely, of the form $\sum_{x\in T}f(x) \neq f(\bigoplus_{x \in T}{x})$ for sets $T \subset \{0,1\}^d$ of even size.

Witnesses of non-quadraticity are even more complicated. 
The tester of Alon et al.~\cite{AlonKKLR05} looks for witnesses consisting of points $x,y,z$, and all four of their linear combinations.
We describe a two-player game that models the interaction between the tester and the adversary and give a winning strategy for the tester-player. We also consider {\em witness structures} in which all specified tuples are witnesses of non-quadraticity (to allow for the possibility of the adversary erasing some points from the structure). We analyze the probability of getting a witness structure under uniform sampling when the input function is $\eps$-far from quadratic.
Our investigation leads to a deeper understanding of the structure of witnesses for both properties, linearity and quadraticity.

In contrast to linearity and quadraticity, we show that several other properties, specifically, sortedness and the Lipschitz property of sequences, and the Lipschitz property of functions $f:\{0,1\}^d\to\{0,1,2\}$, cannot be tested in the presence of an online-erasure oracle, even with $t=1$, no matter how many queries the algorithm makes. Interestingly, witnesses for these properties have a much simpler structure than witnesses for linearity and quadraticity. Consider the case of sortedness  of integer sequences, represented by functions $f:[n]\to\mathbb{N}.$ A sequence is {\em sorted} (or the corresponding function is {\em monotone}) if $f(x) \leq f(y)$ for all $x< y $ in $[n]$.
A witness of non-sortedness consists of two points $x<y$, such that $f(x)>f(y).$
In the standard model, sortedness can be $\eps$-tested with an algorithm that queries $O(\sqrt{n/\eps})$ uniform and independent points~\cite{FLNRRS02}. (The fastest testers for this property have $O(\frac{\log \eps n}\eps)$ query complexity~\cite{EKKRV00,DGLRRS99,BGJRW12,CS13,Belovs18}, but 
they make correlated queries that follow a more complicated distribution.) Our impossibility result demonstrates that even the simplest 
testing strategy of querying independent points
can be thwarted by an online adversary. 
To prove this result, we use sequences that are far from being sorted, but where each point is involved in only one witness, allowing the oracle to erase the second point of the witness as soon as the first one is queried. Using a version of Yao's principle that is suitable for our model, we turn these examples into a general impossibility result for testing sortedness with a 1-online-erasure oracle.

Our impossibility result for testing sortedness uses sequences with many (specifically, $n$) distinct integers. We show that this is not a coincidence by designing a $t$-online-erasure-resilient sortedness tester that works for sequences that have $O(\frac {\eps^2n}{t})$ distinct values. However, the number of distinct values does not have to be large to preclude testing the Lipschitz property in our model. A function $f:[n]\to \mathbb{N}$, representing an $n$-integer sequence, is {\em Lipschitz} 
if $|f(x) - f(y)| \leq |x - y|$ for all $x, y \in [n]$.  Similarly, a function $f \colon \{0, 1\}^d \to \R$ is {\em Lipschitz} if $|f(x) - f(y)| \leq \| x - y \|_1$ for all $x, y \in \{0,1\}^d$. 
We show that the Lipschitz property of sequences, as well as $d$-variate functions, cannot be tested even when the range has size 3, even with $t=1$, no matter how many queries the algorithm makes.

\subparagraph*{Comparison to related models.} Our model is closely related to (offline) erasure-resilient testing of  Dixit et al.~\cite{DixitRTV18}.
In the model of Dixit et al., also investigated in~\cite{RV18,RRV19,BenFLR20,PallavoorRW22,LPRV21,NV20},
 the adversary performs all erasures to the function before the execution of the algorithm. An (offline) erasure-resilient tester is given a parameter $\alpha\in(0,1)$, an upper bound on the fraction of the values that are erased. %
 The adversary we consider is more powerful in the sense that it can perform erasures online, during the execution of the tester. However, in some parameter regimes, our oracle cannot perform as many erasures.
Importantly, all three properties that we show are impossible to test in our model, are testable in the model of Dixit et al.\ with essentially the same query complexity as in the standard model~\cite{DixitRTV18}. 
It is open if there are properties that have lower query complexity in the online model than in the offline model. The models are not directly comparable because the erasures are budgeted differently.

Another widely studied model in property testing is that of tolerant testing~\cite{PRR06}. As explained by Dixit et al., every tolerant tester is also (offline) erasure-resilient with corresponding parameters.
As pointed out in~\cite{PRR06}, the BLR tester is a tolerant tester of linearity for $\alpha$ significantly smaller than $\eps.$ 
Tolerant testing of linearity with distributional assumptions was studied in~\cite{KoppartyS09} and tolerant testing of low-degree polynomials over large alphabets was studied in~\cite{GuruswamiR05}.
Tolerant testing of sortedness is closely related to approximating the distance to monotonicity and estimating the longest increasing subsequence. These tasks can be performed with polylogorithmic in $n$ number of queries~\cite{PRR06,ACCL07,SaksS17}. As we showed, sortedness is impossible to test in the presence of online erasures. 

Adversarial models are also studied in related contexts, such as sampling from  a stream~\cite{Ben-EliezerY20}, more general stream computations \cite{Ben-EliezerJWY22}, and dynamic algorithms~\cite{BeimelKMNSS22}. 
Notably, in these models, the input is formed adversarially online while the algorithm is running; there is no notion of corrupted input. In contrast, in our model, algorithms solve problems on the original, ground-truth input while the access to the input is degrading.

\subsection{Our results}\label{sec:results}

We design 
$t$-online-erasure-resilient testers for linearity and quadraticity, two properties widely studied because of their connection to probabilistically checkable proofs, %
hardness of approximating NP-hard problems, and coding theory. 
Our testers have {\em 1-sided error}, that is, they always accept functions with the property. They are also {\em nonadaptive}, that is, their queries do not depend on answers to previous queries. This is despite the adversary being allowed to respond online (i.e., adaptively) to the algorithm's queries.

\subparagraph*{Linearity.} 

Starting from the pioneering work of \cite{BlumLR93}, linearity testing has been investigated, e.g., in \cite{BellareGLR93, BellareS94, FeigeGLSS96, BellareCHKS96, BellareGS98, Trevisan98, SudanT98,  SamorodnitskyT00, HastadW03,Ben-SassonSVW03, Samorodnitsky07, SamorodnitskyT09, ShpilkaW06, KaufmanLX10} (see~\cite{RasR16} for a survey). 
Linearity can be $\eps$-tested in the standard property testing model with $O(1/\eps)$ queries by the BLR tester. We say that a pair $(x, y)$ \emph{violates} linearity if $f(x) + f(y) \neq f(x \oplus y)$. The BLR tester repeatedly selects 
a uniformly random pair of domain points and rejects if it violates linearity.
A tight lower bound of $\Theta(\eps)$ on the probability that a uniformly random pair violates linearity
was proven by Bellare et al.~\cite{BellareCHKS96} and Kaufman et al.~\cite{KaufmanLX10}.

We show that linearity can be $\eps$-tested with $\widetilde O(\log t/\eps)$ queries with a $t$-online-erasure oracle. 
\begin{theorem}
\label{thm:linearity_min}
There exist a constant $c_0 \in (0,1)$ and a 1-sided error, nonadaptive, $t$-online-erasure-resilient $\eps$-tester for linearity of functions $f\colon \{0,1\}^d \to \{0,1\}$ that  works for all $t \leq c_0\cdot\eps^{5/4} \cdot  2^{d/4}$ and makes $O\big(\min\big( \frac{1}{\eps}\log \frac{t}{\eps}, \frac t \eps\big)\big)$ queries.
\end{theorem}

Our linearity tester has query complexity $O(1/\eps)$ for constant $t$, which is optimal even in the standard property testing model, with no erasures. The tester looks for more general witnesses of non-linearity than the BLR tester, namely,  it looks for tuples $T$ of elements from $\{0,1\}^d$ such that $\sum_{x \in T} f(x) \neq f(\bigoplus_{x\in T} x)$ and $|T|$ is even.
We call such tuples $\emph{violating}$. The analysis of our linearity tester crucially depends on the following structural theorem.

\begin{theorem}\label{thm:fourier}
Let $T$ be a tuple of a fixed even size, where each element of $T$ is sampled uniformly and independently at random from $\{0,1\}^d$. If a function $f \colon \{0,1\}^d \to \{0,1\}$ is $\eps$-far from linear, then 
\begin{equation*}
    \Pr_T\Big[\sum_{x\in T}f(x) \neq f(\bigoplus_{x\in T} x )\Big] \geq \eps.
\end{equation*}
\end{theorem}

\noindent  Our theorem generalizes the result of \cite{BellareCHKS96}, which dealt with the case $|T|=2$. We remark that the assertion in \Thm{fourier} does not hold for odd $|T|$. Consider the function $f(x) = x[1] + 1 \pmod 2$, where $x[1]$ is the first bit of $x$. Function $f$ is $\frac 12$-far from linear, but  has no violating tuples of odd size.

The core procedure of our linearity tester queries $\Theta(\log(t/\eps))$ uniformly random points from $\{0,1\}^d$ to build a reserve and then queries sums of the form $\bigoplus_{x \in T} x$, where $T$ is a uniformly random tuple of reserve elements such that $|T|$ is even. The \emph{quality} of the reserve is the probability that $T$ is violating. The  likelihood that the procedure catches a violating tuple depends on the quality of the reserve (which is a priori unknown to the tester) and the number of sums queried. Instead of querying the same number of sums in each iteration of the core procedure, %
we
obtain a better query complexity by guessing different reserve qualities for each iteration and querying the number of sums that is inversely proportional to the reserve quality. We decide on the number of sums to query based on the {\em work investment} strategy by Berman, Raskhodnikova, and Yaroslavtsev~\cite{BermanRY14}, which builds on an idea proposed by Levin and popularized by Goldreich~\cite{Goldreich14}.

Next, we show that our tester has optimal query complexity in terms of the erasure budget~$t$. 

\begin{theorem}\label{thm:linearity_lower}
For all $\eps \in(0, \frac{1}{4}]$, every $t$-online-erasure-resilient $\eps$-tester for linearity of functions $f \colon \{0,1\}^d\to \{0,1\}$ must make more than $\log_2 t$ queries. %
\end{theorem}

The main idea in the proof of \Thm{linearity_lower} is that when a tester makes $\lfloor\log_2 t\rfloor$ queries, the adversary has the budget to erase all linear combinations of the previous queries after every step. As a result, the tester cannot distinguish a random linear function from a random function.

\subparagraph*{Quadraticity.}
Quadraticity and, more generally, low-degree testing have been studied, e.g.,
in~\cite{BabaiFL91,BabaiFLS91,GemmellLRSW91,FeigeGLSS96,FriedlS95,RubinfeldS96, RazS97, AlonKKLR05, AroraS03, MoshkovitzR08, Moshkovitz17, KaufmanR06, Samorodnitsky07, SamorodnitskyT09, JutlaPRZ09, BKSSZ10,HaramatySS13, Ron-ZewiS13, DinurG13}. 
Low-degree testing is closely related to local testing of Reed-Muller codes. The %
Reed-Muller code ${\cal C}(k,d)$ 
consists of codewords, each of which corresponds to all evaluations of a  polynomial $f\colon \{0,1\}^d \to \{0,1\}$  of degree at most $k$. %
A local tester for a code queries a few locations of a codeword; it accepts if the codeword is in the code; otherwise, it rejects with probability proportional to the distance of the codeword from the code.

In the standard property testing model, quadraticity can be $\eps$-tested with $O(1/\eps)$ queries by the tester of Alon et al.~\cite{AlonKKLR05} that repeatedly selects $x,y,z\sim\{0,1\}^d$ and queries $f$ on all of their linear combinations---the points themselves, the double sums $x\oplus y,x\oplus z, y\oplus z$, and the triple sum $x\oplus y\oplus z$. The tester rejects if the values of the function on all seven queried points sum to 1, since this cannot happen for a quadratic function.
 A tight lower bound on the probability that the resulting 7-tuple is a witness of non-quadraticity was proved by  Alon et al.~\cite{AlonKKLR05}
 and Bhattacharyya et al.~\cite{BKSSZ10}.

We prove that quadraticity can be $\eps$-tested with $O(1/\eps)$ queries with a $t$-online-erasure-oracle for constant~$t$. 
Our tester can be easily modified to give a local tester for the Reed-Muller code $\mathcal{C}(2,d)$ that works with a $t$-online-erasure oracle.

\begin{theorem}\label{thm:quadraticity}
There exists a 1-sided error, nonadaptive, $t$-online-erasure-resilient $\eps$-tester for quadraticity of functions $f\colon \{0,1\}^d \to \{0,1\}$ that makes $O(\frac{1}{\eps})$ queries for constant $t$.
\end{theorem}

 The dependence on $t$ in the query complexity of our quadraticity tester is at least doubly exponential, and it is an open question whether it can be improved.
 The main ideas behind our quadraticity tester are explained in \Sec{techniques}.

\subparagraph*{Sortedness.} Sortedness testing (see~\cite{Enc1} for a survey) was introduced by Ergun et al.~\cite{EKKRV00}.  Its query complexity has been pinned down to $\Theta(\frac{\log (\eps n)}{\eps})$ by~\cite{EKKRV00,Fis04, ChSe14,Belovs18}.

We show that online-erasure-resilient testing of integer sequences is, in general, impossible.
\begin{theorem}\label{thm:sortedness_real}
	For all $\eps \in (0,\frac{1}{12}]$, 
	there is no 1-online-erasure-resilient $\eps$-tester for sortedness of integer sequences.
\end{theorem}
In the case without erasures, sortedness can be tested with $O(\sqrt{n/\eps})$ uniform and independent queries \cite{FLNRRS02}. \Thm{sortedness_real} implies that a uniform tester for a property does not translate into the existence of an online-erasure-resilient tester, counter to the intuition that testers that make only uniform and independent queries should be less prone to adversarial attacks. Our lower bound construction demonstrates that the structure of violations to a property plays an important role in determining whether the property is testable.

The hard sequences from the proof of \Thm{sortedness_real} have $n$ distinct values. Pallavoor et al. \cite{PRV18,Ramesh} considered the setting when the tester is given an additional parameter $r$, the number of distinct elements in the sequence, and obtained an $O(\frac{\log r}\eps)$-query tester. Two lower bounds apply to this setting: $\Omega(\log r)$ for nonadaptive testers~\cite{BlaRY14} and $\Omega(\frac{\log r}{\log \log r})$ for all testers for the case when $r=n^{1/3}$ \cite{Belovs18}. Pallavoor et al.\ also showed that
sortedness can be tested with $O(\sqrt{r}/\eps)$ uniform and independent queries. We extend the result of 
Pallavoor et al.\ to the setting with online erasures for the case when $r$ is small.

\begin{theorem}\label{thm:sortedness_uniform}
Let $c_0 > 0$ be a constant.
There exists a 1-sided error, nonadaptive, $t$-online-erasure-resilient $\eps$-tester for sortedness of $n$-element sequences with at most $r$ distinct values. The tester makes $O(\frac{\sqrt{r}}{\eps})$ uniform and independent queries and works when $r<\frac{\eps^2 n}{c_0 t}$.
\end{theorem}

Thus, sortedness is not testable with online erasures when $r$ is large and is testable in the setting when~$r$ is small. For example, for Boolean sequences, it is testable with $O(1/\eps)$ queries.

\subparagraph*{The Lipschitz property.} Lipschitz testing, introduced by \cite{JhaR13}, was subsequently studied in~\cite{CS13, DixitJRT13, BermanRY14, AwasthiJMR16, ChakrabartyDJS17}. Lipschitz testing of functions %
$f:[n]\to\{0,1,2\}$ can be performed with $O(\frac 1\eps)$ queries~\cite{JhaR13}. For functions $f:\{0,1\}^d\to\mathbb{R}$, it can be done with $O(\frac d\eps)$ queries~\cite{JhaR13,CS13}. 

We show that the Lipschitz property is impossible to test in the online-erasures model even when the range of the function has only 3 distinct values. This applies to both domains, $[n]$ and $\{0,1\}^d.$

\begin{theorem}\label{thm:lipschitz}
    For all $\eps \in (0, \frac{1}{8}]$,
    there is no 1-online-erasure-resilient $\eps$-tester for
    the Lipschitz property of functions $f \colon [n] \to \{0,1,2\}$. 
    The same statement holds when the domain is $\{0,1\}^d$ instead of $[n].$
\end{theorem}

\subparagraph*{Yao's minimax principle.} All our lower bounds use Yao's minimax principle. A formulation of Yao's principle suitable for our online-erasures model 
is described in \Sec{yao}.

\subsection{The ideas behind our quadraticity tester}\label{sec:techniques}
 One challenge in generalizing
the tester
of Alon et al.~\cite{AlonKKLR05} to work with an online-erasure oracle is that 
its queries are correlated. First, we want to ensure that the tester can obtain function values on tuples of the %
form
$(x,y,z,x\oplus y, x\oplus z, y\oplus z, x\oplus y \oplus z)$. %
Then  we  want to ensure that, if the original function is far from the property, the tester is likely to catch such a tuple that is also a witness to not satisfying the property. Next, we formulate a two-player game\footnote{This game has been tested on real children, and they spent hours playing it.} that abstracts the first task. 
In the game, the tester-player sees what erasures are made by the oracle-player. This assumption is made to abstract out the most basic challenge and is not used in the algorithms' analyses.

\subparagraph*{Quadraticity 
testing as a two-player game.} Player 1 represents the tester and Player 2 represents the adversary. The players take turns drawing points, connecting points with edges, and coloring triangles specified by drawn points, each in their own color.
Player 1 wins the game if it draws in blue all the vertices and edges of a triangle and colors the triangle blue. The vertices represent the points $x, y, z \in \{0,1\}^d$, the edges are the sums $x\oplus y, x\oplus z, y\oplus z$, and the triangle is the sum $x\oplus y\oplus z.$
A move of Player~1 consists of drawing a point or an edge between two existing non-adjacent points or coloring an uncolored triangle between three existing points (in blue). A move of Player~2 consists of at most $t$ steps; in each step, it can draw a red edge between existing points or color a triangle between three existing points (in red).

\begin{figure}[ht]
    \begin{center}
    \includegraphics[scale = 0.7]{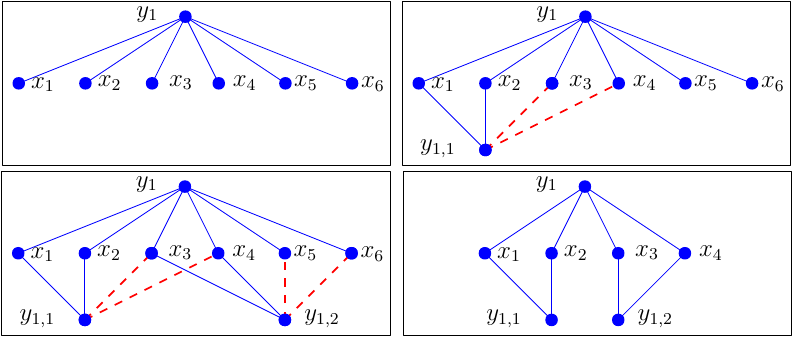}

    \includegraphics[scale = 0.5]{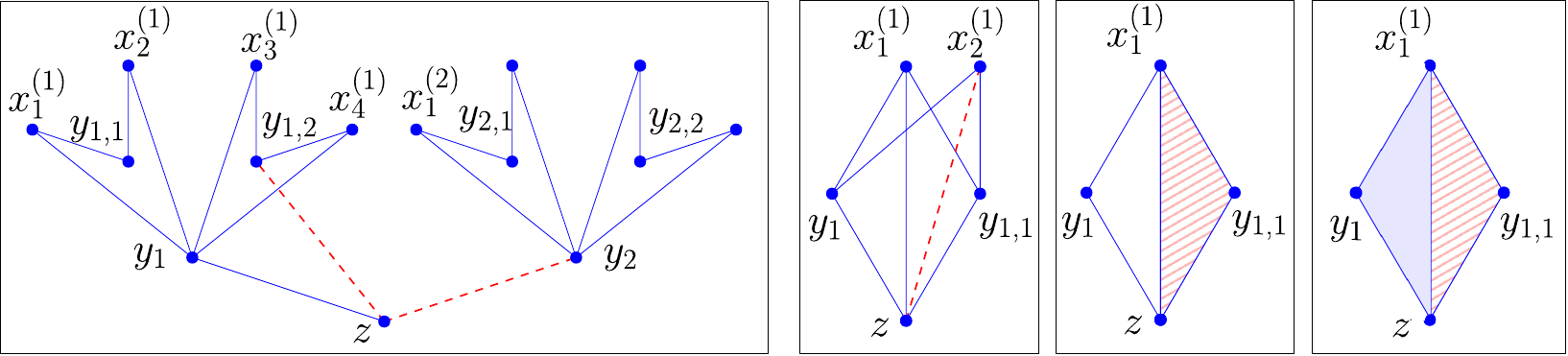}
    \end{center}
\caption{Stages in the quadraticity game for $t=1$, played according to the winning strategy for Player~1: connecting the $y$-decoys from the first tree to $x$-decoys (frames 1-4); drawing %
$z$ and connecting it to $y$-decoys and an $x$-decoy (frames 5-6), and creation of a blue triangle (frames 7--8). Frame 5 contains 
edges from $z$ to two structures, each replicating frame 4. We depict only points and edges %
relevant for subsequent frames.}
\label{fig:quadraticity-game}
\end{figure}

Our online-erasure-resilient quadraticity tester is based on a winning strategy for 
Player~1 with $t^{O(t)}$ moves.  At a high level, Player 1 first draws many decoys for $x$. The $y$-decoys are organized in $t+1$ full $(t+1)$-ary trees of depth $t$. The root for tree $i$ is $y_i$, its children are $y_{i,j}$, where $j\in[t+1]$, etc. %
We jot the rest of the winning strategy for $t=1$ and depict it in \Fig{quadraticity-game}. %
In this case, Player 1 does the following for each of two trees: it draws points $x^{(i)}_1,\dots, x^{(i)}_{12},y_i$; connects $y_i$ to half of $x$-decoys (w.l.o.g., $x^{(i)}_1,\dots,x^{(i)}_6$); draws point $y_{i,1},$
connects it to two of the $x$-decoys adjacent to $y_1$ (w.l.o.g., $x^{(i)}_1$ and $x^{(i)}_2$); draws point $y_{i,2},$ connects it to two of $x^{(i)}_3,\dots,x^{(i)}_6$ (w.l.o.g., $x^{(i)}_3$ and $x^{(i)}_4$); draws $z$ and connects it to one of the roots (w.l.o.g., $y_1$), connects $z$ to one of $y_{1,1}$ and $y_{1,2}$ (w.l.o.g., $y_{1,1}$), connects $z$ to one of $x^{(1)}_1$ and $x^{(1)}_2$ (w.l.o.g., $x^{(1)}_1$), and finally colors one of the triangles $x^{(1)}_1y_1z$ and  $x^{(1)}_1y_{1,1}z$, thus winning the game. The decoys are arranged to guarantee that Player 1 always has at least one available move in each step of the strategy.

For general $t$, the winning strategy is described in full detail in \Alg{quadratic_t}. Recall that the $y$-decoys are organized in $t+1$ full $(t+1)$-ary trees of depth $t$. For every root-to-leaf path in every tree, Player 1 draws edges from all the nodes in that path to a separate set of $t+1$ decoys for $x$. After $z$ is drawn, the tester ``walks'' along a root-to-leaf path in one of the trees, drawing edges between $z$ and the $y$-decoys on the path. The goal of this walk is to avoid the parts of the tree spoiled by Player 2. Finally, Player 1 connects $z$ to an $x$-decoy that is adjacent to all vertices in the path, and then colors a triangle involving this $x$-decoy, a $y$-decoy from the chosen path, and $z$. The structure of decoys guarantees that Player 1 always has $t+1$ options for its next move, only $t$ of which can be spoiled by Player~2.

\subparagraph*{From the game to a tester.}
There are two important aspects of designing a tester that are abstracted away in the game: First, the tester does not actually know which values are erased until it queries them. Second, the tester needs to catch a witness demonstrating a violation of the property, not merely a tuple of the right form with no erasures. Here, we briefly describe how we overcome these challenges.

  Our quadraticity tester is based directly on the game. It converts the moves of the winning strategy of Player~1 into a corresponding procedure, making a uniformly random guess at each step about the choices that remain nonerased. There are three core technical lemmas used in the analysis of the algorithm. \Lem{erasures} lower bounds the probability that the tester makes correct guesses at each step about which edges (double sums) and triangles (triple sums) remain nonerased, thus addressing part of the first challenge. This probability depends only on the erasure budget $t$. To address the second challenge, \Lem{quadratic_t} gives a lower bound on the probability that uniformly random sampled points (the $x$- and $y$- decoys together with $z$) form a large violation structure, where all triangles that Player~1 might eventually complete violate quadraticity. Building on a result of Alon et al., we show that even though the number of triangles involved in the violation structure is large, namely $ t^{O(t)}$, the probability of sampling such a structure is $\alpha = \Omega(\min(\eps, c_t))$, where $c_t \in(0,1)$ depends only on $t$. Finally, \Lem{witness} shows that despite the online adversarial erasures, the tester has a probability of $\alpha/2$ of sampling the points of such a large violation structure and obtaining their values from the oracle. These three lemmas combined show that quadraticity can be tested with $O(1/\eps)$ queries for constant~$t$.   

\subsection{Sublinear-time computation in the presence of online corruptions}

Our model of online erasures extends naturally to a model of online corruptions, where the adversary is allowed to modify values of the input function. Specifically, a \emph{$t$-online-corruption} oracle $\mathcal{O}$, after answering each query, can replace values $\mathcal{O}(x)$ on up to $t$ domain points $x$ with arbitrary values in the range of $f$. Our algorithmic results in the presence of online erasures have implications for computation in the presence of online corruptions. 

We consider two types of computational tasks in the presence of online corruptions. The first one is {\em $t$-online-corruption-resilient} testing, defined analogously to $t$-online-erasure-resilient testing. Specifically, we call an algorithm a {\em $t$-online-corruption-resilient $\eps$-tester} for property~$\cP$ if, given parameters $t\in\mathbb{N}$ and $\eps\in(0,1),$ and access to an input function $f$ via a $t$-online-erasure oracle, the algorithm accepts with probability at least 2/3 if $f\in\cP$ and rejects with probability at least 2/3 if $f$ is $\eps$-far from $\cP$.
In \Lem{no_erasures} we study a general computational task for which the algorithm has oracle access to an input function and has to output a correct answer with high probability. We consider algorithms that access their input via an online-erasure oracle and show that if with high probability they output a correct answer \emph{and} encounter no erased queries during their execution, they also output a correct answer with high probability in the presence of online corruptions. Note that the correctness of the answer is with respect to the original function $f$, as opposed to the corrupted values represented by $\mathcal{O}$.
\Lem{no_erasures}  uses the notion of an \emph{adversarial strategy}.
We model an adversarial strategy as a distribution on decision trees, where each tree dictates the erasures to be made for a given input and queries of the algorithm.

\begin{lemma}\label{lem:no_erasures}
    Let $T$ be an algorithm that is given access to a function via a $t$-online-erasure oracle and performs a specified computational task. Suppose that for all adversarial strategies, with probability at least $2/3$, the algorithm $T$ outputs a correct answer and queries no erased values during its execution. Then, for the same computational task, when $T$ is given access to a function via a $t$-online-corruption oracle, it outputs a correct answer with probability at least $2/3$. 
\end{lemma}

This lemma can be applied to our online-erasure-resilient testers for linearity and for sortedness with few distinct values.  These testers have 1-sided error, i.e., they err only when the function is far from the property. We can amplify the success probability of the testers to $5/6$ for functions that are far from the property, without changing the asymptotic query complexity of the testers. In addition, the resulting testers have a small probability (specifically, at most $1/6$) of encountering an erasure during their computation. As a result, \Lem{no_erasures} implies online-corruption-resilient testers for these properties. Their performance is summarized in \Cors{linearity_corruptions}{sortedness_corrupt}.

For the second type of computational tasks, we return to the motivation of the algorithm looking for evidence of fraud.  In \Lem{detect}, we  show that every nonadaptive, 1-sided error, online-erasure-resilient tester for any property $\cP$ can be modified so that, given access to an input function $f$ that is far from $\cP$ via an online-corruption oracle, it outputs a witness of $f$ not being in $\cP$. 
This result applies to our testers for sortedness with few distinct values, linearity, and quadraticity. Note that the values of the witness could be corrupted values of $f$. For this task, it is helpful to think of the input $f$ as changing dynamically rather than there being a ground truth input $f$ as is the case for \Lem{no_erasures}.

\begin{lemma}\label{lem:detect}
    Let $\mathcal{P}$ be a property of functions that has a 1-sided error, nonadaptive, $t$-online-erasure-resilient tester.
    Then there exists a nonadaptive algorithm with the same query complexity that, given a parameter $\eps$ and access via a $t$-online-corruption-oracle $\cO$ to a function $f$ that is $\eps$-far from $\cP$, outputs with probability at least $2/3$ a witness of  $f \notin \mathcal{P}$. 
    The witness consists of queried domain points $x_1, \dots, x_k$ and values $\cO(x_1), \dots, \cO(x_k)$ obtained by the algorithm, such that no $g \in \cP$ has $g(x_i) = \cO(x_i)$ for all $i \in [k]$. 
\end{lemma}

\subsection{Conclusions and open questions}\label{sec:open}
We initiate a study of sublinear-time algorithms in the presence of online adversarial erasures. We design efficient online-erasure-resilient testers for several important properties (linearity, quadraticity, and---for the case of small number of distinct values---sortedness). For linearity, we prove tight upper and lower bounds in terms of $t$. We also show that several basic properties, specifically, sortedness of integer sequences and the Lipschitz properties, cannot be tested in our model. We now list several open problems.
\begin{itemize}
    \item Sortedness is an example of a property that is impossible to test with online erasures, but is easy to test with offline erasures, as well as tolerantly. Is there a property that has %
    smaller query complexity in the online-erasure-resilient model than in the (offline) erasure-resilient model of~\cite{DixitRTV18}?
    \item We design a $t$-online-erasure-resilient quadraticity tester that makes $O(1/\eps)$ queries for constant $t$. What is the query complexity of $t$-online-erasure-resilient quadraticity testing in terms of $t$ and $\eps$? Specifically, the dependence on $t$ in the query complexity of our quadraticity tester is at least doubly exponential, and it is open  whether it can be improved.
    \item The query complexity of $\eps$-testing if a function is a polynomial of degree at most $k$ is $\Theta(\frac{1}{\eps} + 2^k)$ \cite{AlonKKLR05,BKSSZ10}. Is there a low-degree test for $k\geq 3$ that works in the presence of online erasures?
    \item Our tester for linearity works in the presence of online corruptions, but our tester for quadraticity does not. The reason for this is that our linearity tester is unlikely to see erasures, but that is not the case for our quadraticity tester. Is there an online-corruption-resilient quadraticity tester?
\end{itemize}

\section{An online-erasure-resilient linearity tester}\label{sec:linearity}
To prove \Thm{linearity_min}, we present and analyze two testers. Our main online-erasure-resilient linearity tester (\Alg{linearity_3}) is presented in this section. Its query complexity has optimal dependence on $t$ and nearly optimal dependence on $\eps$. Its performance is summarized in \Thm{linearity_3}. To complete the proof of \Thm{linearity_min}, we give a $O(t/\eps)$-query 
linearity tester in 
\Sec{linearity_appendix}. 
It has optimal query complexity for constant~$t$.

\begin{theorem}\label{thm:linearity_3}
There exists %
$c_0 \in (0,1)$ and a 1-sided error, nonadaptive,~$t$-online-erasure-resilient $\eps$-tester for linearity of functions~$f\colon \{0,1\}^d \to \{0,1\}$ that  works for $t \leq c_0\cdot\eps^{5/4} \cdot  2^{d/4}$ and makes $O(\frac{1}{\eps}\log \frac{t}{\eps})$ queries.
\end{theorem}

The $t$-online-erasure-resilient tester guaranteed by \Thm{linearity_3} is presented in \Alg{linearity_3}.

\begin{algorithm} 
	\caption{An Online-Erasure-Resilient Linearity Tester}
	\begin{algorithmic}[1]
		\Require{$\eps \in (0,\frac{1}{2})$, erasure budget $t\in\mathbb{N}$, access to function $f: \{0,1\}^d \rightarrow \{0,1\}$ via a $t$-online-erasure oracle.} 
		
		\State Let $q = 2\log\frac{50t}\eps$. \label{step:define_q}
		\State\textbf{for all} $j \in [\log \frac 8 \eps]$: \label{step:define_j}
		
		\State\label{step:repeat} \indent \textbf{repeat} $\frac{8 \ln 5}{2^j\eps}$ times:
 		
		\State \indent \indent \textbf{for all} {$i \in [q]$}:\label{step:loop_query_x_i}

		\State \indent \indent \indent Sample $x_i \sim \{0,1\}^d$ and query $f$ at $x_i$. \label{step:query_x_i}
		
		\State \indent \indent\textbf{repeat} $4 \cdot 2^j$ times:\label{step:query_sums}

        \State \indent \indent \indent Sample a uniform nonempty subset $I$ of $[q]$ of even size.

        \State \indent \indent \indent Query $f$ at $\bigoplus_{i \in I}{x_i}$.\label{step:sample-I}
		
		\State \indent \indent \indent {\textbf{Reject}} if $\sum_{i\in I}f(x_i) \neq f(\bigoplus_{i \in I}{x_i})$ and all points $x_i$ for $i\in I$ and $\bigoplus_{i \in I}{x_i}$ are nonerased.  \label{step:reject}

		\State \textbf{Accept.}
	\end{algorithmic}
	\label{alg:linearity_3}
\end{algorithm}

\subsection{Proof of \Thm{fourier}}\label{sec:fourier}

 In this section, we prove \Thm{fourier}, the main structural result used in \Thm{linearity_3}. Recall that a $k$-tuple $(x_1, \dots, x_k) \in (\{0,1\}^d)^k$ violates linearity if $f(x_1) + \dots + f(x_k) \neq f(x_1 \oplus \dots \oplus x_k)$. (Addition is mod $2$ when adding values of Boolean functions.) \Thm{fourier} states that if $f$ is $\eps$-far from linear, then for all even $k$, with probability at least $\eps$, independently and uniformly sampled points $x_1, \dots, x_k \sim \{0,1\}^d$ form a violating $k$-tuple. 
Our proof of \Thm{fourier} 
builds on the proof of \cite[Theorem 1.2]{BellareCHKS96}, which is a special case of \Thm{fourier} for $k=2$. The proof is via Fourier analysis. Next, we state some standard facts and definitions related to Fourier analysis. See, e.g., \cite{OD1} for proofs of these facts. 

Consider the space of all real-valued functions on $\{0,1\}^d$ equipped with the inner-product 
\[\langle g, h \rangle = \underset{x\sim \{0,1\}^d}{\E}[g(x)h(x)],\]
where $g, h \colon \{0,1\}^d \to \R$. The \emph{character} functions $\chi_S \colon \{0, 1\}^d \to \{-1,1\}$, defined as $\chi_S = (-1)^{\sum_{i \in S} x[i]}$ for $S \subseteq [d]$, form an orthonormal basis for the space of functions under consideration. Hence, every function $g:\{0,1\}^d\to\R$ can be uniquely expressed as a linear combination of the functions $\chi_S,$ where $S\subseteq [d]$. The Fourier coefficients of $g$ are the coefficients on the functions  $\chi_S$ in this linear representation of $g$.  

\begin{definition}[Fourier coefficient]
For $g \colon \{0,1\}^d \to \R$ and $S \subseteq [d]$, the Fourier coefficient of $g$ on $S$ is 
$
    \widehat{g}(S) = \langle g, \chi_S \rangle = \underset{x\sim \{0,1\}^d}{\E}[g(x)\chi_S(x)].
$
\end{definition}

We will need the following facts about Fourier coefficients. 

\begin{theorem}[Parseval's Theorem]\label{thm:parseval}
For all functions $g\colon\{0,1\}^d \to \R$, it holds that $ \langle g,g\rangle =  \sum_{S\subseteq [d]} \widehat{g}(S)^2$.
In particular, if $g\colon\{0,1\}^d \to \{-1,1\}$ then $ \sum_{S\subseteq [d]} \widehat{g}(S)^2=1$. 
\end{theorem}

\begin{theorem}[Plancherel's Theorem] \label{thm:plancherel}
For all functions $g, h \colon\{0,1\}^d \to \R$, it holds that $ \langle g, h\rangle = \sum_{S\subseteq [d]}  \widehat{g}(S)\widehat{h}(S).$
\end{theorem}

A function $g\colon\{0,1\}^d \to \{-1,1\}$ is linear if $g(x)g(y) = g(x\oplus y)$ for all $x, y \in \{0,1\}^d$. 

\begin{lemma}\label{lem:dist_to_linearity}
The distance of $g\colon\{0,1\}^d \to \{-1,1\}$ to linearity is $\frac{1}{2} - \frac{1}{2}\max_{S\subseteq [d]}\widehat{g}(S)$. 
\end{lemma}

Finally, we also use the convolution operation, defined below, and one of its key properties. 
 
\begin{definition}[Convolution]
Let $g, h \colon\{0,1\}^d \to \R$. Their convolution is the function $g *h: \{0,1\}^d \to \R$ defined by
$
    (g * h)(x) = \underset{y\sim \{0,1\}^d}{\E}[g(y)h(x\oplus y)].
$
\end{definition}

\begin{theorem} \label{thm:convolution}
Let $g, h \colon\{0,1\}^d \to \R$. Then, for all $S \subseteq [d],$ it holds $\widehat{g * h}(S) = \widehat{g}(S)\widehat{h}(S)$. 
\end{theorem}

\begin{proof}[Proof of \Thm{fourier}]
    Define $g \colon \{0,1\}^d \to \{-1, 1\}$ so that $g(x) = (-1)^{f(x)}$.
   That is, $g$ is obtained from the function $f$ by encoding its output with $\pm 1$.  Note that the distance to linearity of $g$ is the same as the distance to linearity of $f$. 
   We have that the expression $\frac{1}{2} - \frac{1}{2}g(x_1)\dots g(x_k)g(x_1 \oplus \dots \oplus x_k)$~is an indicator for the event that 
   points $x_1, \dots, x_k\sim \{0,1\}^d$ violate linearity 
   for $g$. Define $g^{*k}$ to be the convolution of $g$ with itself $k$ times, i.e., $g^{*k} = g * \dots * g$, where the $*$ operator appears $k-1$ times.  We obtain
  \begin{align}
      &\Pr_{x_1, \dots, x_k \sim \{0,1\}^d}[(x_1, \dots, x_k) \textnormal{ violates linearity}]\nonumber \\
      &= \underset{x_1, \dots, x_k \sim \{0,1\}^d}{\E} \Big[\frac{1}{2} - \frac{1}{2}g(x_1)\dots g(x_k)g(x_1 \oplus \dots \oplus x_k)\Big] \nonumber \\
      &= \frac{1}{2} - \frac{1}{2} \underset{x_1 , \dots, x_{k-1} \sim \{0,1\}^d}{\E} [g(x_1) \dots g(x_{k-1}) \cdot \underset{x_k \sim \{0,1\}^d}{\E} [g(x_k) g(x_1 \oplus \dots \oplus x_k)]]
      \nonumber  \\
      &= \frac{1}{2} - \frac{1}{2} \underset{x_1 , \dots, x_{k-1} \sim \{0,1\}^d}{\E} [g(x_1) \dots g(x_{k-1}) (g * g)(x_1\oplus \dots \oplus x_{k-1})]
      \label{eq:definition} \\
      &= \frac{1}{2} - \frac{1}{2} \underset{x_1 \sim \{0,1\}^d}{\E} [g(x_1) \cdot (g^{*k})(x_1)]  \label{eq:repeat} \\
      &=  \frac{1}{2} - \frac{1}{2}\sum_{S \subseteq [d]} \widehat{g}(S) \widehat{g^{*k}}(S) \label{eq:plancherel}\\ 
      &=  \frac{1}{2} - \frac{1}{2} \sum_{S \subseteq [d]} \widehat{g}(S)^{k+1}, \label{eq:convolution}
  \end{align}
  where  \Eqn{definition} holds by the definition of convolution, \Eqn{repeat} follows by repeated application of the steps used to obtain \Eqn{definition},  \Eqn{plancherel} follows from Plancherel’s Theorem (\Thm{plancherel}), and  \Eqn{convolution} follows from \Thm{convolution}.
  
  Note that $|\widehat{g}(S)| \leq 1$ for all $S \subseteq [d]$, because $\widehat{g}(S)$ is the inner product of two functions with range in $\{-1, 1\}$. In addition,  $\widehat{g}(S) \geq 0$ for $S$ such that $\chi_S$ is the closest linear function to $g$. Then, for even $k$,
  \begin{equation*}
      \sum_{S \subseteq [d]} \widehat{g}(S)^{k+1} \leq \max_{S \subseteq [d]} (\widehat{g}(S)^{k-1}) \sum_{S \subseteq [d]} \widehat{g}(S)^{2} = \max_{S \subseteq [d]} (\widehat{g}(S)^{k-1}) \leq \max_{S \subseteq [d]} \widehat{g}(S),
  \end{equation*}
where the equality follows from Parseval's Theorem (\Thm{parseval}). By \Lem{dist_to_linearity}, the distance of $g$ to linearity is $\frac{1}{2} - \frac{1}{2} \max_{S \subseteq [d]} \widehat{g}(S)$, 
which is at least $\eps$, since $g$ is $\eps$-far from linear. This concludes the proof.
\end{proof}

\subsection{Proof of \Thm{linearity_3}} \label{sec:linearity_proof}

In this section, we prove \Thm{linearity_3} using \Thm{fourier}. In \Lem{distinct_and_nonerased}, we analyze the probability of good events that capture, roughly, that the queries made in the beginning of each iteration haven't already been ``spoiled'' by the previous erasures.
Then we use the work investment strategy of \cite{BermanRY14}, stated in \Lem{work}, 
together with \Thm{fourier} and \Lem{distinct_and_nonerased} to prove \Thm{linearity_3}.

Each iteration of the outer \textbf{repeat} loop in Steps~\ref{step:repeat}-\ref{step:reject} of \Alg{linearity_3} is called a {\em round}. We say a query $x$ is {\em successfully obtained} if it is nonerased when queried, i.e., the tester obtains $f(x)$ as opposed to $\perp$.

\begin{lemma}[Good events]
\label{lem:distinct_and_nonerased}
Fix one round of \Alg{linearity_3}. Consider the points  $x_1,\dots, x_q$ queried in Step~\ref{step:query_x_i} of this round, where $q = 2\log(50t/\eps)$,
and the set $S$ of all sums $\bigoplus_{i\in I}x_i,$ where $I$ is a nonempty subset of $[q]$ of even size.
Let $G_1$ be the (good) event that all points in $S$
are distinct. Let $G_2$ be the (good) event that all points $x_1,\dots,x_q$
are successfully obtained and all points in $S$
are nonerased at the beginning of the round. Finally, let $G = G_1 \cap G_2$. Then $\Pr[\overline{G}] \le \frac{\eps}{2}$ for all adversarial strategies.
\end{lemma}
\begin{proof}
First, we analyze event $G_1$. Consider points $x_{i_1}, \dots, x_{i_{k}}, x_{i_1'}, \dots, x_{i_{\ell}'} \sim \{0,1\}^d$,  where $\{i_1, \dots, i_k\} \neq \{i_1', \dots, i_\ell'\}$, $k, \ell \in [q]$ and $k, \ell$ are even. Since the points are distributed uniformly and independently, so are the sums $x_{i_1} \oplus\dots\oplus x_{i_{k}}$ and $x_{i_1'}\oplus\dots\oplus x_{i_{\ell}'}$. The probability that two uniformly and independently sampled points $x,y \sim \{0,1\}^d$ are identical is $\frac{1}{2^d}$. The number of sets $I \subseteq [q]$ of even size is $2^{q-1}$ because every subset of $[q-1]$ can be uniquely completed to such a set $I$.
By a union bound over all pairs of sums, $\Pr[\overline{G_1}] 
\leq \frac{2^{2q}}{4 \cdot 2^d}$. 

To analyze $G_2,$ fix any adversarial strategy. The number of queries made by \Alg{linearity_3} is at most 
\begin{equation}\label{eq:num-queries}
    \sum_{j=1}^{ \log(8/\eps) } \frac{8\ln 5}{2^j \eps} \Big(q + 4 \cdot 2^j \Big) \leq \frac{8\ln 5}{\eps} q + \frac{32 \ln 5}{\eps} \log\frac 8\eps \leq \frac{8\ln 5}{\eps} q + \frac{16 \ln 5}{\eps}q \leq \frac{40 q}{\eps }.
\end{equation}
Hence, the oracle erases at most $ \frac{40qt}{\eps}$ points. Since each point $x_i$ is sampled uniformly from $\{0,1\}^d$,
\begin{equation*}
    \Pr_{x_i \sim \{0,1\}^d} [x_i \textnormal{ is erased when queried}] \leq  \frac{40qt}{\eps\cdot 2^d}.
\end{equation*}
Additionally, before the queries $x_i$ are revealed to the oracle, each sum $\bigoplus_{i \in I}x_i$ is distributed uniformly at random. Therefore, for every $\{i_1, \dots, i_k\} \subseteq [q]$,
\begin{equation*}
    \Pr_{x_{i_1}, \dots,  x_{i_{k}} \sim \{0,1\}^d} [x_{i_1} \oplus \dots \oplus x_{i_{k}} \textnormal{ is erased at the beginning of the round}] \leq \frac{40qt}{\eps\cdot 2^d}.
\end{equation*}
By a union bound over the $q$ points sampled in Step~\ref{step:query_x_i} and at most $2^{q-1}$ sums, we get 
\[
\Pr[\overline{G_2}] \leq  \frac{40qt}{\eps  2^d}(q+2^{q-1}) \leq  \frac{40  qt \cdot 2^q }{\eps \cdot   2^d}.
\]
Since $q = 2\log\frac{50t}\eps$, we get  $\frac{40qt}{\eps} \leq \frac{3}{4} \cdot 2^q$ and, consequently,
\begin{align*}
     \Pr[\overline{G}] \leq \frac{2^{2q}}{4 \cdot 2^d} + \frac{40qt\cdot 2^q }{\eps \cdot  2^d}  
     \leq \frac{ 2^{2q}}{2^d} \leq \frac{ 50^4 \cdot t^4}{\eps^4 2^d} \leq \frac{50^4 \cdot c_0^4 \cdot \eps^5}{\eps^4} \leq \frac{\eps}{2},
\end{align*}
since $t \leq c_0 \cdot \eps^{5/4} \cdot 2^{d/4}$, as stated in \Thm{linearity_3}, and assuming $c_0$ is sufficiently small.
\end{proof}

Next, we state the work investment lemma.
\begin{lemma}[Lemma 2.5 of \cite{BermanRY14}] \label{lem:work}
Let $X$ be a random variable taking values in $[0,1]$. Suppose $\E[X] \geq \alpha$. Let $s = \lceil \log(\frac{4}{\alpha}) \rceil$ and $\delta \in (0,1)$ be the desired probability of error. For all $j \in [s]$, let $p_j = \Pr[X \geq 2^{-j}]$ and $k_j =  \frac{4\ln(1/\delta)}{2^j\alpha}$. Then
$
    \prod_{j=1}^{s}(1-p_j)^{k_j} \leq \delta.
$
\end{lemma}

\begin{proof}[Proof of \Thm{linearity_3}]
By \Eqn{num-queries}, the query complexity of \Alg{linearity_3} is $O(\frac q \eps)=O(\frac{\log (t/\eps)}\eps)$. \Alg{linearity_3} is nonadaptive and always accepts if $f$ is linear. Suppose now that $f$ is $\eps$-far from linear and fix any adversarial strategy. We show that~\Alg{linearity_3} rejects with probability at least $2/3$. 

 Consider the \emph{last} round of \Alg{linearity_3}. For points $x_1, \dots, x_q \sim \{0,1\}^d$ sampled in Step~\ref{step:query_x_i} of this last round, let $Y$ denote the fraction of nonempty sets $\{i_1, \dots, i_{k}\} \subseteq [q]$ such that $k$ is even and $(x_{i_1}, \dots, x_{i_{k}})$ violates linearity. Recall the event $G$ defined in \Lem{distinct_and_nonerased}. Let  $\mathbf{1}_{G}$ be the indicator random variable for the event $G$ for the last round.

\begin{claim}\label{clm:x}
Let $X = Y \cdot \mathbf{1}_{G}$, where $Y$ is as defined above. Then $\E[X] \geq \frac{\eps}{2}$. 
\end{claim}
\begin{proof}
For all nonempty $\{i_1, \dots, i_k \} \subseteq [q]$, such that $k$ is even, let $Y_{i_1,\dots, i_{k}}$ be the indicator
for the event that $(x_{i_1}, \dots, x_{i_{k}})$ violates linearity. By \Thm{fourier} and the fact that $k$ is even, 
\[\E[Y_{i_1,\dots, i_{k}}] = \Pr_{x_{i_1},\dots,  x_{i_{k}} \sim \{0,1\}^d}[ (x_{i_1}, \dots, x_{i_{k}})\textnormal{ violates linearity}] \geq \eps.\] 
We obtain a lower bound on $\E[Y]$ by linearity of expectation.
\begin{equation}\label{eq:y}
    \E[Y]= \frac{1}{2^{q-1} - 1}\sum_{\{i_1, \dots, i_k\} \subseteq [q],k \textnormal{ even}} \E[Y_{i_1,\dots, i_{k}}] \geq \eps.
\end{equation}

Observe that $X=Y$ when $G$ occurs, and $X=0$ otherwise. By the law of total expectation,
\begin{align*}
 \E[X] &=\E[X|G]\cdot\Pr[G]+ \E[X|\overline{G}]\cdot\Pr[\overline{G}] 
 = \E[X|G]\cdot\Pr[G]=\E[Y|G]\cdot\Pr[G] \\
&=\E[Y]- \E[Y|\overline{G}]\cdot\Pr[\overline{G}]\geq \eps - 1\cdot (\eps/2) = \eps/2,
\end{align*}
where the inequality follows from  \Eqn{y}, the fact that $Y\leq 1$, and \Lem{distinct_and_nonerased}. 
\end{proof}

Fix any round of \Alg{linearity_3} and the value of $j$ used in this round (as defined in Step~\ref{step:define_j}). Let $X'$ be defined as $X$, but for this round instead of the last one.
The round is \emph{special} if $X' \geq 2^{-j}$. Let $p'_j = \Pr[X' \geq 2^{-j}]$ and $p_j = \Pr[X\geq 2^{-j}]$. 
Then  $p'_j\geq p_j$, since the number of erasures only increases with each round. For each $j$, there are $k_j = \frac{8 \ln 5}{2^j \eps}$ rounds of \Alg{linearity_3} that are run with this particular value of $j$. Since \Alg{linearity_3} uses independent random coins for each round, the probability that no round is special is at most
\begin{equation*}
    \prod_{j=1}^{\log \frac{8}{\eps}}(1-p'_j)^{k_j} \leq \prod_{j=1}^{\log \frac{8}{\eps}}(1-p_j)^{k_j} \leq \frac{1}{5},
\end{equation*}
where the last inequality follows by \Lem{work} applied with $\delta=1/5$ and $\alpha=\eps/2$ and \Clm{x}. Therefore, with probability at least $\frac{4}{5}$, \Alg{linearity_3} has a special round.

Consider a special round of \Alg{linearity_3} and fix the value of $j$ for this round. We show that \Alg{linearity_3} rejects in the special round 
with probability at least $5/6$. We call a sum $\bigoplus_{i\in I}x_i$ {\em violating} if the tuple $(x_i \colon i \in I)$ violates linearity. Since $G$ occurred, all $q$ points queried in Step~\ref{step:query_x_i} of \Alg{linearity_3} were successfully obtained. So, the algorithm will reject as soon as it successfully obtains a violating sum. Since $G$ occurred, there are at least $2^{q-1}-1$ distinct sums that can be queried in Step~\ref{step:sample-I}, all of them nonerased at the beginning of the round. \Alg{linearity_3} makes at most $q + 4\cdot 2^j$ queries in this round, and thus the fraction of these sums erased during the round is at most
\begin{align*}
    t\cdot\frac{q + 4\cdot 2^j}{2^{q-1}-1}
   &\leq  t\cdot \Big(\frac{1}{2^{q/2}} + \frac{3\cdot 2^j}{2^{q-2}} \Big)
   \leq 
  \frac{t}{t \cdot \frac{50}{\eps}} + \frac{12t \cdot 2^j}{t^2 \cdot (\frac{50}{\eps})^2} \\
   &=   \frac{8}{50 \cdot \frac{8}{\eps}} + \frac{12 \cdot 8^2\cdot 2^j}{50^2t \cdot (\frac{8}{\eps})^2} 
   \leq 
 \frac{0.16}{2^j} + \frac{0.3072}{ 2^j} 
      \leq 
     \frac{1}{2 \cdot 2^j},
\end{align*}
where in the first inequality we used that $\frac{q}{2^{q-1}-1} \leq \frac{1}{2^{q/2}}$  for $q \geq 9$ and that $2^{q-1}-1\geq (4/3)\cdot 2^{q-2}$ for $q\geq 3$ (note that $q \geq 2\log(50t/\eps)\geq 2\log 100>13$), in the second inequality we used $q \geq 2\log(50t/\eps)$, and in the third inequality we used $2^j \leq \frac{8}{\eps}$ and $t\geq 1$. 

Since the round is special, at least a $2^{-j}$ fraction of the sums that can be queried in Step~\ref{step:sample-I} are violating.
Thus,
the fraction of the sums that are violating and nonerased before each iteration of Steps~\ref{step:sample-I}-\ref{step:reject} in this round is at least $2^{-j}-2^{-j-1}=2^{-j-1}.$

Then, each iteration of Steps~\ref{step:sample-I}-\ref{step:reject} rejects with probability at least $2^{-j-1}.$ Since there are $4\cdot 2^j$ iterations with independently chosen sums, the probability that the special round accepts is at most 
\begin{align*}
(1-2^{-j-1})^{4 \cdot 2^j} \leq \eee^{-2}\leq  1/ 6.
\end{align*}
That is, the probability that \Alg{linearity_3} rejects in the special round is at least $5/6$. Since the special round exists with probability at least $\frac{4}{5}$, \Alg{linearity_3} rejects with probability at least $\frac{4}{5}\cdot \frac{5}{6} = \frac{2}{3}$. 
\end{proof}

\section{An online-erasure-resilient linearity tester with $O(t/\eps)$ queries} \label{sec:linearity_appendix}

In this section, we present our online-erasure-resilient linearity tester with query complexity $O(t/\eps)$ (\Alg{linearity}) and prove  \Thm{linearity} below. \Thm{linearity} together with \Thm{linearity_3} implies \Thm{linearity_min}. 

\begin{theorem}\label{thm:linearity}
There exist a constant $c_0 \in (0,1)$ and a 1-sided error, nonadaptive, $t$-online-erasure-resilient $\eps$-tester for linearity of functions $f\colon \{0,1\}^d \to \{0,1\}$ that  works for $t \leq c_0\cdot\eps \cdot  2^{d/4}$ and makes $O(\frac{t}{\eps})$ queries.
\end{theorem}

\begin{algorithm} %
	\caption{An Online-Erasure-Resilient Linearity Tester}
	\begin{algorithmic}[1]
		\Require{$\eps \in (0,\frac{1}{2})$, erasure budget $t\in\mathbb{N}$, access to function $f: \{0,1\}^d \rightarrow \{0,1\}$ via a $t$-online-erasure oracle.}
		
		\State Let $c = 88$ and $q = \:\frac{ct}{\eps}$. 
		
		\State\textbf{for all} {$i \in [q]$}:\label{step:loop_query_x_i_2}

		\State  \indent Sample $x_i \sim \{0,1\}^d$ and query $f$ at $x_i$. \label{step:query_x_i_2}
		
		\State\textbf{repeat} $\frac{24}{\eps}$ times:\label{step:query_sum}
		
		\State  \indent Sample a uniform $(i, j)$ such that $i,j\in[q]$ and $i <j$, then  query $f$ at $x_i \oplus x_j$.\label{step:sample-ij}
		
		\State  \indent {\textbf{Reject}} if $x_i, x_j$, and $x_i\oplus x_j$ are nonerased and $f(x_i) + f(x_j) \neq f(x_i\oplus x_j)$.\label{step:check_I}

		\State \textbf{Accept.}
	\end{algorithmic}
	\label{alg:linearity}
\end{algorithm}

In the analysis of \Alg{linearity}, we let $x_i$, for all $i \in [q]$, be random variables denoting the points queried in~Step~\ref{step:query_x_i_2}. 
Recall that a pair $(x, y)$ of domain points  \emph{violates} linearity if $f(x) + f(y) \neq f(x \oplus y)$.
The proof of \Thm{linearity} crucially relies on the following lemma about the concentration of pairs violating linearity. %

\begin{lemma}\label{lem:violating_sums}
Suppose $f$ is $\eps$-far from linear and $q = \frac{ct}{\eps}$. Let $Y$ denote the number of pairs $(i,j) \in [q]^2$ such that $i < j$ and $(x_i,x_j)$ violates linearity. Then $\Pr[Y \leq \frac{\eps}{4}\cdot \binom{q}{2}]\leq 0.25$.  %
\end{lemma}

We prove \Lem{violating_sums} in \Sec{violating_sums}. In the rest of this section, we first prove an auxiliary lemma (\Lem{distinct-and-nonerased}) and then use Lemmas~\ref{lem:violating_sums}~and~\ref{lem:distinct-and-nonerased} to prove \Thm{linearity}.

We say a query $u \in \{0,1\}^d$ is {\em successfully obtained} if it is nonerased right before it is queried, i.e., the tester obtains $f(u)$ as opposed to the symbol~$\perp$.

\begin{lemma}\label{lem:distinct-and-nonerased}
Let $G_1$ be the (good) event that all $\binom{q}{2}$ sums $x_i \oplus x_j$ of points queried in Step~\ref{step:query_x_i_2} of \Alg{linearity} are distinct, and $G_2$ be the (good) event that all $q$ points $x_i$ queried in Step~\ref{step:query_x_i_2} are successfully obtained. Then $\Pr[\overline{G_1} \cup \overline{G_2}] \le 0.05$ for all adversarial strategies.
\end{lemma}
\begin{proof}
First, we analyze event $G_1$. Consider points $x_i, x_j, x_k, x_\ell \sim \{0,1\}^d$,  where $\{i,j\} \neq \{k, \ell\}$.  Since these points are distributed uniformly and independently, so are the sums $x_i \oplus x_j$ and $x_k \oplus x_\ell$. The probability that two uniformly and independently sampled points $x,y \sim \{0,1\}^d$ are identical is  $\frac{1}{2^d}$. Therefore, by a union bound over all pairs of sums, $\Pr[\overline{G_1}] \leq \binom{q}{2}^2 \cdot \frac{1}{2^d} \leq \frac{q^4}{4\cdot 2^d}.$ 

To analyze $G_2,$ fix an arbitrary adversarial strategy.
In Steps~\ref{step:loop_query_x_i_2}-\ref{step:query_x_i_2}, the algorithm makes at most $q$ queries. Hence the oracle erases at most $qt$ points from $\{0,1\}^d$. Since each point $x_i$ is sampled uniformly from $\{0,1\}^d$,
\begin{equation*}
    \Pr_{x_i \sim \{0,1\}^d} [x_i \textnormal{ is erased when queried}] \leq \frac{qt}{ 2^d}.
\end{equation*}
By a union bound over the $q$ points sampled in Step~\ref{step:query_x_i_2}, we get $\Pr[\overline{G_2}] \leq \frac{q^2t}{2^d}$. We substitute $q = \frac{ct}{\eps}$ and get
\begin{align*}
     \Pr[\overline{G_1} \cup \overline{G_2}] \leq \frac{q^4}{4 \cdot 2^d} + \frac{q^2t}{2^d} \leq \frac{q^4}{2^d} = \frac{c^4 t^4}{\eps^4 \cdot 2^d} \leq  c^4c_0^4 \leq 0.05,
\end{align*}
since $t \leq c_0\cdot\eps \cdot  2^{d/4}$, as stated in \Thm{linearity}, and assuming $c_0$ is sufficiently small.
\end{proof}

\begin{proof}[Proof of \Thm{linearity}] The query complexity of \Alg{linearity} is evident from its description. Clearly, \Alg{linearity} is nonadaptive and accepts if $f$ is linear. Suppose now that $f$ is $\eps$-far from linear and fix an arbitrary adversarial strategy. We show that~\Alg{linearity} rejects with probability at least $\frac{2}{3}$.

We call a sum $x_i \oplus x_j$ queried in Step~\ref{step:sample-ij} {\em violating} if the pair $(x_i, x_j)$ violates linearity. For \Alg{linearity} to reject, %
it must sample $i,j\in[q]$ such that $x_i, x_j,$ and $x_i \oplus x_j$ are successfully obtained and $x_i \oplus x_j$ is a violating sum. 
Let $E_k,$ for all $k\in[24/\eps],$ be the event that in iteration $k$ of Step~\ref{step:query_sum}, the algorithm successfully obtains a violating sum. Let $E=\bigcup_{k\in[24/\eps]} E_k$. 
We define the following good event concerning the execution of the \textbf{for} loop: $G=G_1\cap G_2\cap G_3,$ where $G_1$ and $G_2$ are as defined in \Lem{distinct-and-nonerased}, and $G_3$ is the event that $Y$, defined in \Lem{violating_sums}, is at least $\frac{\eps} 4 \cdot\binom{q}{2}$. By a union bound and Lemmas~\ref{lem:violating_sums}~and~\ref{lem:distinct-and-nonerased}, we know that $\Pr[G]\geq 0.7.$
Therefore, the probability that \Alg{linearity} rejects 
   is at least 
   \begin{align*}
         \Pr[E \cap G]=\Pr[G]\cdot\Pr[E|G]\geq 0.7\cdot\Pr[E|G].
   \end{align*}

Suppose that $G$ occurs for the execution of the \textbf{for} loop. Since $G_2$ occurred, all $q$ points queried in Step~\ref{step:query_x_i_2} of \Alg{linearity} were successfully obtained. So, the algorithm will reject as soon as it successfully obtains a violating sum. Since $G_1$ occurred, there are at least $\binom{q}{2}$ distinct sums that can be queried in Step~\ref{step:sample-ij}. Since $G_3$ occurred, at least $\frac \eps 4\cdot \binom{q}{2}$ of them are violating.
The total number of erasures the adversary makes over the course of the entire execution is at most $2qt$, since \Alg{linearity} queries at most $2q$ points.
Therefore, the fraction of sums that are nonerased and violating before each iteration of Step~\ref{step:sample-ij} is at least
\begin{align*}
\frac \eps 4 -\frac {2qt}{\binom q 2}
\geq\frac \eps 4 -\frac {5t}{q}
=\frac \eps 4 -\frac {5\eps}{c}
=\frac \eps 4 -\frac {5\eps}{88}
> \frac \eps 6. 
\end{align*}
That is, $\Pr[E_k|G]> \eps/6$ for all $k\in[24/\eps].$
Observe that the indicator for the event $E_k$, for every fixed setting of random coins (of the algorithm and the adversary) used prior to iteration $k$ of Step~\ref{step:sample-ij}, 
is a Bernoulli random variable with the probability parameter equal to the fraction of sums that are violating and nonerased right before iteration $k$. %
Since independent random coins are used in each iteration of Step~\ref{step:sample-ij},
\begin{align*}
   \Pr[\overline{E}|G]
\leq \Pr\big[\bigcap_{k\in[24/\eps]}\overline{E_k}\mid G\big]
< (1-\eps/6)^{24/\eps}
\leq \eee^{-4}. 
\end{align*}
This yields that $\Pr[E\cap G]>0.7\cdot(1-\eee^{-4})> 2/3$, completing the proof of the theorem.
\end{proof}

\subsection{Concentration of pairs that violate linearity}\label{sec:violating_sums}

In this section, we prove \Lem{violating_sums} on the concentration of pairs that violate linearity. %
The proof relies on  \Clm{witness_upper} that gives upper bounds on the fraction of pairs $(x,y) \in (\{0,1\}^d)^2$ that violate linearity and on the fraction of triples $(x, y, z) \in \{0,1\}^{d\times 3}$ such that $(x,y)$ and $(x,z)$ violate linearity. The upper bounds are stated as a function of the distance to linearity.
The distance of a function $f$ to linearity, denoted $\eps_f$, is the minimum over all linear functions $g$ over the same domain as $f$ of 
$\Pr_x[f(x)\neq g(x)]$.

 We remark that a tighter upper bound than \Eqn{witness_upper} in \Clm{witness_upper} is shown in \cite[Lemma 3.2]{BellareCHKS96}. Since we only need a looser upper bound, we provide a proof for completeness. 

\begin{claim}\label{clm:witness_upper}
For all $f\colon \{0,1\}^d \to \{0,1\}$, the following hold:
\begin{align}
    \Pr_{x,y\sim \{0,1\}^d}[(x,y) \textnormal{ violates linearity}] &\leq 3\eps_f; \label{eq:witness_upper}\\
        \Pr_{x,y,z\sim \{0,1\}^d}[(x,y) \textnormal{ and } (x, z) \textnormal{ violate linearity}] &\leq \eps_f + 4\eps_f^2.\label{eq:double_witness_upper} 
\end{align}
\end{claim}
\begin{proof}
Let $g$ be the closest linear function to $f$. Let $S \subseteq \{0,1\}^d$ be the set of points on which $f$ and $g$ differ, i.e., $S = \{x \in \{0,1\}^d \: | \: f(x)\neq g(x)\}$. Then $|S| = \eps_f \cdot 2^d$.

For a pair $(x, y)$ to violate linearity, at least one of $x, y$, and  $x\oplus y$ must be in $S$. By a union bound,
\begin{align*}
  \Pr_{x,y \sim\{0,1\}^d}[(x,y)
  \textnormal{ violates linearity for $f$}] 
 &\leq \Pr_{x,y\sim \{0,1\}^d}[x \in S \vee y \in S \vee x\oplus y \in S]\\ 
     &\leq 3\Pr_{x\sim\{0,1\}^d}[x \in S] = 3\eps_f,
\end{align*}
completing the proof of \Eqn{witness_upper}. 

To prove \Eqn{double_witness_upper}, let $E_{x,y,z}$ be the event that $(x,y)$ and $(x,z)$ violate linearity. By the law of total probability,
\begin{align*}   %
      \Pr_{x,y,z\sim \{0,1\}^d}&[E_{x,y,z}] \\
      &=      \Pr[E_{x,y,z} \: | \: x\in S]\Pr[x \in S]+ \Pr[E_{x,y,z} \: | \: x\notin S]\Pr[x \notin S] \nonumber \\
      &\leq 1\cdot \eps_f  + \Pr[E_{x,y,z} \: | \: x\notin S] \\ 
      &\leq \eps_f + \Pr[(y,z \in S) \vee (x \oplus y, z \in S) \vee (y, x \oplus z \in S) \vee (x \oplus y, x \oplus z \in S)\mid x \notin S]\\
      &\leq \eps_f + 4\Pr[y,z \in S\mid x \notin S] 
      = \eps_f + 4\eps_f^2,
\end{align*}
where, in the last line, we used symmetry and a union bound, and then independence of $x, y$, and $z$. 
\end{proof}

\begin{proof}[Proof of \Lem{violating_sums}]
We prove the lemma by applying Chebyshev's inequality to $Y$.
For all $i, j\in[q]$, where $i < j,$ let $Y_{i,j}$ be the indicator 
for the event that $(x_i, x_j)$ violates linearity.
Then $Y = \sum_{i < j \in [q]} Y_{i,j}$. 

First, we obtain a lower bound on $\E[Y]$.  By linearity of expectation and symmetry,
\begin{equation*}
    \E[Y]= \sum_{i < j \in [q]} \E[Y_{i,j}] 
    = \binom{q}{2}\cdot \Pr_{x_i, x_j \sim \{0,1\}^d}[(x_i, x_j) \textnormal{ violates linearity}] \geq \binom{q}{2}\eps_f,
\end{equation*}
where the inequality is due to \cite[Theorem 1.2]{BellareCHKS96}.

Next, we obtain an upper bound on $\Var[Y]$. We have
\[
  \Var[Y] = \sum_{i < j} \Var[Y_{i,j}] + \sum_{\substack{i<j,k<l \\ \{i,j\} \neq \{k,\ell\}   }} \Cov[Y_{i,j}Y_{k,l}].
\]
By~\Eqn{witness_upper}, we get $\Var[Y_{i,j}] 
\leq \E[Y^2_{i,j}]=\E[Y_{i,j}] \leq 3\eps_f$. Therefore,
\begin{equation}\label{eq:variance}
    \sum_{i < j} \Var[Y_{i,j}] \leq 3\binom{q}{2}\eps_f.
\end{equation}
When $\{i,j\} \cap \{k,\ell\} = \emptyset$, the random variables $Y_{i,j}$ and $Y_{k,\ell}$ are independent, and thus $\Cov[Y_{i,j}Y_{k,l}] = 0$. Consider the case when $|\{i,j\} \cap \{k,\ell\}| = 1$. Suppose that $i = k$. Then
\begin{align*}
    \Cov[Y_{i,j}Y_{i,\ell}] &= \E[Y_{i,j}Y_{i,\ell}]-\E[Y_{i,j}]\E[Y_{i,\ell}]
    \leq \eps_f + 4\eps_f^2 - \eps_f^2 = \eps_f + 3\eps_f^2,
\end{align*}
where the inequality follows from \Eqn{double_witness_upper} and \cite[Theorem 1.2]{BellareCHKS96}. By symmetry,

\begin{equation}\label{eq:covariance}
  \sum_{\substack{i<j,k<l \\ \{i,j\} \neq \{k,\ell\}   }} \Cov[Y_{i,j}Y_{k,l}] = 6 \sum_{\substack{i < j < \ell}} \Cov[Y_{i,j}Y_{i,l}] \leq 6\binom{q}{3}(\eps_f + 3\eps_f^2) \leq 15\binom{q}{3} \eps_f,
\end{equation}
where the last inequality holds since $\eps_f \leq \frac{1}{2}$. 
Combining \Eqn{variance} with \Eqn{covariance}, we get that
$\Var[Y] \leq 18\binom{q}{3}\eps_f.$

Finally, we use $\eps_f \geq \eps$, our lower bound on $\E[Y],$ and Chebyshev's inequality:
\begin{align*}  %
\Pr \bigg[Y \leq \binom{q}{2} \frac{\eps}{4} \bigg] &\leq \Pr \bigg[Y \leq \binom{q}{2}\frac{\eps_f}{4} \bigg] \leq  
                                                      \Pr\bigg[Y \leq \frac{\E[Y]}{4}\bigg]\\
                                                    & \leq \Pr\bigg[\bigg|Y -\E[Y]\bigg| \geq \frac 3 4 \cdot{\E[Y]}\bigg]
                                                    \leq \frac{16}9\cdot \frac{\Var[Y]}{\E[Y]^2} 
    \leq \frac {16\cdot 18} 9 \cdot \frac{\binom{q}{3}\eps_f}{\binom{q}{2}^2\eps_f^2} \\
    &= \frac {16\cdot 18} 9\cdot \frac 2 3 \cdot \frac{q-2}{q(q-1) \eps_f} 
  < \frac{22}{q\cdot \eps_f} = \frac{22\eps}{ct\eps_f} \leq \frac{22}{ct} \leq \frac{1}{4},
\end{align*}
since $q= \frac{ct}{\eps}$ and  $c=88$. This completes the proof of \Lem{violating_sums}.
\end{proof}

\section{A lower bound for online-erasure-resilient linearity testing}\label{sec:linearity_lower}

In this section, we prove \Thm{linearity_lower} that 
shows that every $t$-online-erasure-resilient $\eps$-tester for linearity of functions $f \colon \{0,1\}^d\to \{0,1\}$ must make more than $\log_2 t$ queries.

\begin{proof}[Proof of \Thm{linearity_lower}] 
The proof is via Yao's minimax principle for the online-erasures model (stated in \Thm{yao} and \Cor{yao_2}).  
 Let $\cD^+$ be the uniform distribution over all linear Boolean functions on $\{0,1\}^d$ and  $\cD^{-}$ be the uniform distribution over all Boolean functions functions on $\{0,1\}^d$. 
 
 We show that a function $f \sim \cD^-$ is $\frac 1 4$-far from linear with probability at least $6/7$.
 Let $g$ be a linear function, $f \sim \cD^-$, and $\dist(f,g)$ be the fraction of domain points on which $f$ and $g$ differ. Then, $\E[\dist(f,g)] = \frac{1}{2}$. By the Hoeffding bound, $\Pr_{f\sim \cD^-}[\dist(f,g) \leq \frac{1}{4}] \leq \eee^{-\frac{2^d}{8}}$.
By a union bound over the $2^d$ linear functions, $\Pr_{f\sim \cD^-}[f \text{ is } \frac{1}{4}\textnormal{-far from linear}] \geq 1 - 2^d \cdot \eee^{-\frac{2^d}{8}}$. For $d$ large enough, this probability is at least $6/7$. 
 
We fix the following  strategy for a $t$-online-erasure oracle $\cO$: after responding to each query, erase $t$ sums of the form $\bigoplus_{x \in T}x$, where  $T$ is a subset of the queries made so far, choosing the subsets $T$ in some fixed order.
If at most $\log_2 t$ queries are made, the adversary erases all the sums of queried points.

 Let $A$ be a deterministic algorithm that makes $q \leq \log_2 t$ queries to the oracle $\cO$. Assume w.l.o.g.\ that $A$ does not repeat queries. 
 We describe two random processes $\cP^+$ and $\cP^-$ that interact with %
 algorithm $A$ in lieu of oracle $\cO$ and provide query answers consistent with a random function from $\cD^+$ and $\cD^-$, respectively. 
 For each query of $A$, both processes $\cP^+$ and $\cP^-$ return $\perp$ if the value has been previously erased by $\cO$; otherwise, they return $0$ or $1$ with equal probability. Thus, the distribution over query-answer histories when $A$ interacts with $\cP^+$ is the same as when $A$ interacts with $\cP^-$.  
 
Next, we describe how the processes $\cP^+$ and $\cP^-$ assign values to the locations of $f$ that were either not queried by $A$ or erased when queried, and show that they generate $\cD^+$ and $\cD^-$, respectively. 
After $A$ finishes its queries, $\cP^-$ sets the remaining unassigned locations (including the erased locations) of the function to be $0$ or $1$ with equal probability.  Clearly, $\cP^-$ generates a function from the distribution $\cD^-$. 
 
 To describe $\cP^+$ fully, first let $Q \subseteq \{0,1\}^d$ denote the queries of $A$ that are answered with a value other than $\perp$. Since $q \leq \log_2 t$, by our choice of the oracle $\cO$, the sum of any subset of vectors in $Q$ is not contained in $Q$. 
Hence, the vectors in $Q$ are linearly independent. Then,  $\cP^+$ completes $Q$ to a basis $B$ for $\{0,1\}^d$ and sets the value of $f$ on all vectors in $B\setminus Q$ independently to $0$ or $1$ with equal probability.

 Since $B$ is a basis, each vector $y \in \{0,1\}^d$ can be expressed as a linear combination of vectors in $B$
 (with coefficients in $\{0,1\}$), that is, $y=\bigoplus_{x \in T}x$ for some $T\subseteq B$.
 The process $\cP^+$ sets $f(y) =\sum_{x \in T}f(x)$, where addition is mod 2. 
 The function $f$ is linear and agrees with all values previously assigned by $\cP^+$ to the vectors in $Q$. Moreover, $f$ is distributed according to  $\cD^+$, since one can obtain a uniformly random linear function by first specifying a basis for $\{0,1\}^d$, and then setting the value of $f$ to be $0$ or $1$ with equal probability for each basis vector. 
 
Thus, $\cP^+$ generates linear functions, $\cP^-$ generates functions that are $\frac 14$-far from linear with probability at least $\frac 67$, and the query-answer histories for any deterministic algorithm $A$ that makes at most $\log_2 t$ queries and runs against our $t$-online-erasure oracle $\cO$ are identical under $\cP^+$ and $\cP^-$. Consequently,
 \Cor{yao_2} implies the desired lower bound.
\end{proof}

\section{An online-erasure-resilient quadraticity tester}

In this section, we state our online-erasure-resilient quadraticity tester (\Alg{quadratic_t}) and prove \Thm{quadraticity}. 
The main idea behind \Alg{quadratic_t} and its representation as a two-player game appear in \Sec{techniques}, accompanied by explanatory figures for the case when $t=1$. 
We now give a high level overview of \Alg{quadratic_t}.
For a function $f \colon \{0,1\}^d \to \{0,1\}$ and $x,y,z \in \{0,1\}^d,$ let 
\begin{align*}
T_f(x, y, z) = \sum_{\emptyset \neq S \subseteq \{x,y,z\}} f\bigl( \bigoplus_{u\in S} u \bigr), 
\end{align*}
where the first sum is mod $2$. We say a triple $(x, y, z)$ \emph{violates} quadraticity if $T_f(x,y,z) = 1$. 
The tester of Alon et al.~samples three vectors $x,y,z \in \{0,1\}^d$ uniformly and independently at random and rejects if $(x,y,z)$ violates quadraticity. Our tester looks for the same kind of violations as the tester of Alon et al.

The main challenge in the design of our online-erasure-resilient tester is to ensure it can  query all three such points and their sums in the presence of a $t$-online-erasure adversary. In each iteration of the \textbf{repeat} loop of \Alg{quadratic_t}, the following steps are performed. 
For each $\ell \in [t+1]$, we first query a \emph{reserve} of uniformly and independently sampled points  $x_i^{(\ell)}$ for $i \in [(t+1)^2(2t+1)^t]$. 
Next, for each $\ell \in [t+1]$, we query a set of points that we visualize as being the nodes of a $(t+1)$-ary tree of depth $t$. 
There is a one-to-one correspondence between the nodes of such a tree and vectors of length up to $t + 1$ over the alphabet $[t+1]$.
For $m \in [t]$, we represent by $y_{(\ell,j_1,\dots,j_m)}$, the sampled point visualized as a node at depth $m$ in the $\ell$-th tree, where the $j_i$'s specify the unique path from the root to that node in the tree.
Now, for $\ell \in [t+1]$, and for each node $y_{\bj}$ in the $\ell$-th tree, where $\bj$ is shorthand for $(\ell,j_1,\dots,j_m)$, we associate with that node a subset $S_{\bj}$ of points from the reserve, and query the points $y_{\bj} \oplus x$ for each $x$ in $S_{\bj}$. The set $S_{\bj}$ is a subset of a specified cardinality of the set $S_{\bj^{(-1)}}$, where $S_{\bj^{(-1)}}$ is the set associated with the parent node of $y_{\bj}$ in the $\ell$-th tree.
Finally, the algorithm queries a point $z$ sampled uniformly and independently at random from $\{0,1\}^d$, and samples a uniformly random leaf $y_{\bj}$ of a uniformly random tree $\ell \in [t+1]$. 
The set $S_{\bj}$ associated with the leaf $y_{\bj}$ has, by construction, $t+1$ points in it. 
All the points in $S_{\bj}$, again, by construction, also belong to the sets $S_{\bj'}$ associated with the nodes $y_{\bj'}$ on the path from the root to the leaf $y_{\bj}$ of the $\ell$-th tree. 
Our algorithm queries $y_{\bj'}\oplus z$ for all such nodes $y_{\bj'}$.
It then samples $x$ uniformly at random from $S_{\bj}$ and queries $x \oplus z$. Finally, it samples a uniformly random node $y_{\bj'}$ on the path from the root to the leaf $y_{\bj}$ and queries $x \oplus y_{\bj'} \oplus z$.
Observe that, by design, the point $x \oplus y_{\bj'}$ has already been queried in an earlier step.
The algorithm rejects if all the points
involved in the sum $T_f(x, y_{\bj'}, z)$  are nonerased and the triple $(x, y_{\bj'}, z)$ violates quadraticity.

\Alg{quadratic_t} uses the following notation.
For a vector $\bj = (\ell, j_1, \dots, j_m)$, where $m \in [0,t]$ and $\ell, j_1, \dots, j_t \in [t+1]$, we use the convention that $\bj = (\ell)$ for $m=0$. For $ k \in[0,m]$, let $\bj^{(k)} = (\ell, j_1, \dots, j_k)$ be the vector containing the first $k+1$ entries of~$\bj$. Finally, let $\bj^{(-1)}$ be the vector $\bj$ with its last entry removed. If $\bj = (\ell)$, then $\bj^{(-1)}$ is the empty vector $\emptyset$.

\begin{algorithm} %
	\caption{A $t$-Online-Erasure-Resilient Quadraticity Tester}
	\begin{algorithmic}[1]
	
		\Require{$\eps \in (0,1)$, access to function $f: \{0,1\}^d \rightarrow \{0,1\}$ via a $t$-online-erasure oracle.}
		
		\State Let $I = (t+1)^{2}(2t+1)^t$, $J = \frac{(t+1)^{(t+1)}-1}{t}$, and  $\alpha =  \min \bigl( \frac{\eps}{2}, \frac{7}{(18 \cdot IJ(t+1))^{I J (t+1)}} \bigr).$
		\label{step:def_constants}
		
		\State\textbf{repeat} $4c_t/\alpha$ times: \Comment{\textsf{$c_t$ is a constant from \Lem{erasures} that depends only on $t$.}}
		
		\State \indent \textbf{for all} $\ell \in [t+1]$:
		
		\State \indent \indent Query $f$ at independent $x_1^{(\ell)}, \dots, x_I^{(\ell)} \sim \{0,1\}^d$, and let $S_\emptyset = \{x_1^{(\ell)}, \dots, x_I^{(\ell)}\}$. \label{step:query_y}
		
		\State \indent \indent   \textbf{for all} integer $m \in [0,t]$:

		\State \indent  \indent \indent   \textbf{for all} $(j_1, j_2, \dots, j_m) \in [t+1]^{m}$: \Comment{\textsf{When $m = 0$, the loop is run once.}} \label{step:query_x}
		
		\State \indent \indent \indent  \indent Let $\bj = (\ell, j_1, \dots, j_m)$ and query $f$ at $y_{\bj}  \sim \{0,1\}^d$. 
\Comment{\textsf{$\bj = (\ell)$ for $m = 0$.}}		
		\State \indent \indent \indent  \indent Let $S_{\bj}$ be a uniformly random subset of $S_{\bj^{(-1)}}$ of size $(t+1)(2t+1)^{t-m}$. \label{step:random_subset}
		
		\State \indent \indent \indent  \indent Query $x \oplus y_{\bj}$ for all $x \in S_{\bj}$. \label{step:query_xy}
		
		\State \indent \indent \indent  \indent Remove $S_{\bj}$ from $S_{\bj^{(-1)}}$.

		\State \indent Sample $z \sim \{0,1\}^d$ and query $f$ at $z$.  \label{step:query_z}
		
		\State \indent Sample $\bj = (\ell, j_1, \dots, j_{t}) \sim [t+1]^{(t+1)}$ and query $f$ at $y_{\bj^{(m)}}   \oplus z$ for all $m \in [0,t]$. \label{step:query_indices}

		\State \indent Suppose $S_{\bj}   = \{x_1^{(\ell)}, x_2^{(\ell)}, \dots, x_{t+1}^{(\ell)}\}$.  Sample $i \sim [t+1]$ and query $f$ at $x_i^{(\ell)} \oplus z$.

		\label{step:query_zy}
		
		\State \indent Sample integer $m \sim [0,t]$ and query $f$ at $x_{i}^{(\ell)} \oplus y_{\bj^{(m)}}   \oplus  z$.\label{step:3sum}

		\State \indent \textbf{Reject} if $T_f(x_{i}^{(\ell)}, y_{\bj^{(m)}} , z) = 1$. \Comment{\textsf{All points needed for computing $T_f$ are nonerased.}}
		
		\State \textbf{Accept}. 
	\end{algorithmic}
	\label{alg:quadratic_t}
\end{algorithm}

\begin{proof}[Proof of \Thm{quadraticity}]

If $f$ is quadratic, then \Alg{quadratic_t} always accepts. Suppose that $f$ is $\eps$-far from quadratic. Fix an adversarial strategy and a round of \Alg{quadratic_t}. We show that \Alg{quadratic_t} rejects with probability  $\Omega(\eps)$ in this round.  

 Observe that \Alg{quadratic_t} makes queries of three types: {\em singletons} (of the form $x_i^{(\ell)},y_{\bj^{(m)}}$, and $z$), {\em doubles} (of the form $x_i^{(\ell)}\oplus y_{\bj^{(m)}}$,$x_i^{(\ell)} \oplus z$, and $y_{\bj^{(m)}}\oplus z$), and {\em triples} (of the form $x_i^{(\ell)}\oplus y_{\bj^{(m)}}\oplus z$), where $i \in [I], \bj = (\ell, j_1, \dots, j_t) \in [t+1]^{(t+1)}, m \in [0, t]$.  We call points of the form $x_i^{(\ell)}\oplus y_{\bj^{(m)}}$, $x_i^{(\ell)} \oplus z$, and $y_{\bj^{(m)}}\oplus z$ {\em double decoys}, and points of the form $x_i^{(\ell)}\oplus y_{\bj^{(m)}} \oplus z$ {\em triple decoys}. We refer to double and triple decoys simply as decoys. Only some of the decoys become actual queries.

Let $\mathcal{G}$ denote the {\em good} event that, for the fixed round, all of the following hold:
\begin{itemize}
    \item all singleton queries are successfully obtained; 
    \item all double decoys of the form $x_i^{(\ell)} \oplus y_{\bj^{(m)}}$ are nonerased right before $y_{\bj^{(m)}}$ is queried;
    \item all double and triple decoys involving $z$ (as well as $x_i^{(\ell)}$ and/or $y_{\bj^{(m)}}$) are nonerased right before $z$ is queried.
\end{itemize}

To lower bound the rejection probability of the algorithm, we consider the event that all triples of the form  $(x_i^{(\ell)}, y_{\bj^{(m)}}, z)$, where $i \in [I], \bj = (\ell, j_1, \dots, j_t) \in [t+1]^{(t+1)}, m \in [0, t]$, violate quadraticity, and all queries in the round are successfully obtained. 
 \begin{definition}[Witness]
The singleton queries form a \emph{witness} if $T_f(x_i^{(\ell)},y_{\bj^{(m)}}, z) = 1$ for all $i \in [I], \bj = (\ell, j_1, \dots, j_t) \in [t+1]^{(t+1)}, m \in [0, t]$, and, in addition, all singletons and all decoys are distinct. 
 \end{definition}
 Let $W$ be the event that the singleton queries form a witness. Let $\alpha$ be as defined in Step~\ref{step:def_constants}. 

 \begin{lemma}[Probability of successful singleton queries]\label{lem:witness} If $f\colon \{0,1\}^d\to\{0,1\}$ is $\eps$-far from being quadratic, then $\Pr[W \cap \mathcal{G}] \geq \alpha/2$. 
\end{lemma}

In other words, \Lem{witness} shows that for every adversarial strategy, with probability at least  $\frac{\alpha}{2}$,  the tester successfully obtains singleton queries that form a witness and, in addition, right before each singleton is queried, the decoys involving that singleton are nonerased. The proof of \Lem{witness} (in \Sec{lem_witness}) relies on the following key structural result about the fraction of large structures where all triples of a certain form violate quadraticity. %

\begin{lemma}[Probability of singletons forming a violation structure]\label{lem:quadratic_t}
    Let $I,J,t\in{\mathbb{N}}$.
	Suppose $f \colon \{0,1\}^d\to \{0,1\}$ is $\eps$-far from being quadratic. For points $x_i^{(\ell)}, y_{j}^{(\ell)}, z \sim \{0,1\}^d$, where $(i,j,\ell) \in [I]\times[J]\times[t+1]$,
 	\begin{equation}\label{eq:violation}
 	    	\Pr\biggl[ \bigcap_{\substack{i \in [I], j\in [J], \ell \in [t+1]}} [T_f(x_i^{(\ell)}, y_{j}^{(\ell)}, z) = 1] \biggr] \geq \alpha, 
 	\end{equation}
 	where $\alpha= \min \bigl( \frac{\eps}{2}, \frac{7}{(18 \cdot IJ(t+1))^{I J (t+1)}} \bigr).$
\end{lemma}

We prove \Lem{quadratic_t} in \Sec{quadratic_witness}, building on a result of \cite{AlonKKLR05}. To prove \Lem{witness}, we use \Lem{quadratic_t} with $I = (t+1)^2(2t+1)^t$ and $J = \sum_{m=0}^{t}(t+1)^m = \frac{(t+1)^{(t+1)}-1}{t}$, which is the number of nodes in each tree.

Next, in \Lem{erasures}, we show that the probability of successfully obtaining all double and triple queries in the round, given the event $W\cap \mathcal{G}$, depends only on $t$. The proof of \Lem{erasures} appears in  \Sec{lem_erasures}.
 \begin{lemma}[Probability of all double and triple queries being nonerased]\label{lem:erasures}
The probability that all queries made in one round of \Alg{quadratic_t} are successfully obtained, conditioned on the event $W \cap \mathcal{G}$, is at least $1/c_t$, where $c_t$ depends only on $t$.
\end{lemma}

Let $H$ denote the event that all queries in the round are successfully obtained, and recall that $W$ is the event that the singleton queries form a witness. The probability that \Alg{quadratic_t} rejects in the given round is at least $\Pr[W \cap H]$. Note that
\begin{equation*}
    \Pr[W \cap H] \geq \Pr[W \cap H \cap \mathcal{G}] = \Pr[H \:|\: W \cap \mathcal{G}]\Pr[W \cap \mathcal{G}].
\end{equation*} By~\Lem{witness}~and~\Lem{erasures}, we have that \Alg{quadratic_t} rejects with probability at least $\frac{1}{c_t} \cdot \frac{\alpha}{2}$ for a fixed round. Thus, after $\frac{4c_t}{\alpha}$ rounds, \Alg{quadratic_t} rejects with probability at least~$\frac{2}{3}$. 
\end{proof}

\subsection{Proof of \Lem{witness}} \label{sec:lem_witness}

In this section, we prove \Lem{witness} that bounds from below the probability that the singleton queries are successfully obtained and form a witness, and, in addition, right before each singleton is queried, the decoys involving that singleton are nonerased. We show that if the fraction of large violation structures is at least~$\alpha$ (\Lem{quadratic_t}), then, despite the adversarial erasures, the probability of successful singleton queries, as described above, is at least $\alpha/2$.

\begin{proof}[Proof of \Lem{witness}]
    
The key idea of the proof is to keep track of \emph{active} witnesses during the execution of the round. To simplify notation, let $K = (I + J)(t+1) + 1$ and denote by $u_1, \dots, u_{K}$ the singleton queries made by \Alg{quadratic_t} in the given round.  

\begin{definition}[Active witness] For $i \in [K]$, let $u_1, \dots, u_i$ denote the singleton queries made by \Alg{quadratic_t} %
up until
a given timepoint of the fixed round. We say a witness $(v_1,\dots, v_{K}) \in (\{0,1\}^d)^{K}$ is  {\em $i$-active} if
\begin{enumerate}
    \item The first $i$ entries of the tuple are equal to $u_1, u_2, \dots, u_i$.
    \item All decoys involving $u_j$, where $j \leq i$, were nonerased right before $u_j$ was queried. 
\end{enumerate}
 
\end{definition}

Furthermore, let $A_i$ be a random variable denoting the number of active witnesses right after the $i$-th singleton query, and let $B_i$ denote the number of active witnesses right before the $i$-th singleton query. Let $B_1$ denote the number of witnesses at the beginning of the round, such that all singletons and decoys for the witness are nonerased. Note that $\Pr[W \cap \mathcal{G}] = \Pr[A_{K} = 1] = \E[A_{K}]$. To lower bound $\E[A_{K}]$, we first show a lower bound on $B_1$, obtained in turn by a lower bound on the total fraction of witnesses. We then bound the difference between $A_i$ and $B_{i+1}$ and show a relationship between $\E[A_i]$ and $\E[B_i]$ for general $i$. All expectations in this proof are over the choice of singletons $u_1, \dots, u_K \sim \{0,1\}^d$. 

\begin{claim} \label{clm:quadratic_witness}
Let $\alpha$ be as defined in \Lem{quadratic_t}.
If $f$ is $\eps$-far from being quadratic then
\begin{equation*}
    \Pr_{u_1, \dots,u_{K} \sim \{0,1\}^d }[(u_1,\dots,u_{K}) \text{ is a witness}] \geq \frac{7\alpha}{8}. 
\end{equation*}
\end{claim}
\begin{proof}
    The probability that the event of \Lem{quadratic_t} is true for the singleton queries $u_1, \dots, u_{K}$ is at least $\alpha$.  Let $S$ denote the set of all singletons and decoys associated with the $K$-tuple $(u_1, \dots, u_{K})$. The number of triple decoys associated with the $K$-tuple is at most $IJ(t+1)$, whereas the number of double decoys associated with the $K$-tuple is at most $I(J+1)(t+1)$. Therefore, $|S| \leq 3IJ(t+1)$. Any two elements in $S$ are uniformly and independently distributed. Thus, the probability that two elements in $S$ are identical is $\frac{1}{2^d}$. By a union bound over all pairs of elements in $S$, the probability that not all elements in $S$ are distinct is at most $\frac{|S|^2}{2^d}$. For $d$ large enough, we have $\frac{|S|^2}{2^d} \leq \frac{\alpha}{8}$. Hence, with probability at least $\alpha - \frac{\alpha}{8} = \frac{7\alpha}{8}$, the singleton queries $u_1, \dots, u_{K}$ form a witness. 
\end{proof}

\begin{claim} For all adversarial strategies, $B_1 \geq \frac{3\alpha}{4} \cdot 2^{Kd}$.  \label{clm:b1}
\end{claim}
\begin{proof}
    Recall that $B_1$ counts the number of witnesses at the beginning of the round, such that  all singletons and decoys for the witness are nonerased. \Alg{quadratic_t} makes at most $|S|$ queries in each round, so it makes at most $\frac{4c_t|S|}{\alpha}$ queries in total. Therefore, the oracle can erase at most $\frac{4t\cdot |S| \cdot c_t}{\alpha}$ points from $\{0,1\}^d$. Each erasure can deactivate at most $|S| \cdot 2^{(K-1)d}$ witnesses, since the point erased could be any one of the $|S|$ singletons and decoys associated with the witness. By \Clm{quadratic_witness}, we obtain $B_1 \geq \frac{7\alpha}{8} \cdot 2^{Kd} - \frac{4t \cdot |S|^2\cdot c_t}{\alpha} \cdot 2^{(K-1)d}$. Since $|S| \leq 3IJ(t+1)$, for $d$ large enough, we have $\frac{4t \cdot |S|^2 \cdot c_t}{\alpha} \cdot 2^{(K-1)d} \leq \frac{\alpha}{8} \cdot 2^{Kd}$. Therefore, $B_1 \geq \frac{3\alpha}{4}\cdot 2^{Kd}$. 
\end{proof}

\begin{claim} \label{clm:difference}
For all $i \in [K-1]$ and all adversarial strategies, 
\begin{align*}
    A_i - B_{i+1} \leq (t+1)^2(2t+1)^t \cdot |S| \cdot 2^{(K-1-i)d}.
\end{align*}
\end{claim}
\begin{proof}
Observe that \Alg{quadratic_t} makes at most $(t+1)(2t+1)^t$ decoy queries between any two singleton queries (this is the number of double queries made in Step~\ref{step:query_xy} for $m=0$).  Therefore, in the period right after the $i$-th singleton query and before the $(i+1)$-st singleton query, the oracle can perform at most $((t+1)(2t+1)^t+1)\cdot t \leq (t+1)^2(2t+1)^t$ erasures. Let $u_1, u_2, \dots, u_i$ be the first $i$ singleton queries of the algorithm. For each erasure, the oracle can deactivate at most $|S|\cdot 2^{(K-1-i)d}$ witnesses whose first $i$ entries are $u_1, u_2, \dots, u_i$, by erasing one of the remaining singletons or decoys. Therefore, the number of active witnesses between right after the $i$-th singleton query and right before the $(i+1)$-st singleton query can decrease by at most $(t+1)^2(2t+1)^t \cdot |S| \cdot 2^{(K-1-i)d}$. 
\end{proof}

\begin{claim} For all $i\in [K]$, it holds that $\E[A_i] = \frac{1}{2^d}\E[B_i]$. 
\label{clm:expectation}
\end{claim}
\begin{proof} For $v \in \{0,1\}^d$, let $B_{i,v}$ denote the number of  witnesses that are active right before the $i$-th singleton query and whose $i$-th entry is equal to $v$. Then, $B_i = \sum_{v \in \{0,1\}^d} B_{i,v}$. Let $\mathbf{1}(v)$ be the indicator random variable for the event that the $i$-th singleton query is equal to~$v$. Then,  
\[
    \E[A_i] =   \E\Big[\sum_{v \in \{0,1\}^d} B_{i,v}\mathbf{1}(v)\Big]
    = \E[\mathbf{1}(v)]\E\Big[\sum_{v \in \{0,1\}^d} B_{i,v}\Big] 
    = \frac{1}{2^d}\E[B_i]\,. 
    \qedhere
\]
\end{proof}

We combine the claims above to complete the proof of the lemma. By \Clm{expectation},
\begin{align*}
    \E[A_{K}] &= \frac{1}{2^d}\E[B_{K}] = \frac{1}{2^d}\E[A_{K-1} + B_{K} - A_{K-1}] = \frac{1}{2^d}\E[A_{K-1}] - \frac{1}{2^d}\E[A_{K-1}-B_{K}] \\
    &= \frac{1}{2^{2d}}\E[B_{K-1}] - \frac{1}{2^d}\E[A_{K-1}-B_{K}] 
    = \dots = \frac{1}{2^{Kd}}\E[B_1] - \sum_{i=1}^{K-1} \frac{\E[A_i - B_{i+1}]}{2^{(K-i)d}}.
\end{align*}
By \Clm{b1},
\begin{equation*}
     \frac{1}{2^{Kd}}\E[B_1] \geq \frac{3\alpha}{4}. 
\end{equation*}
In addition, \Clm{difference} yields that
\begin{align*}
    \sum_{i=1}^{K-1} \frac{\E[A_i - B_{i+1}]}{2^{(K-i)d}} &\leq  \sum_{i=1}^{K-1}\frac{(t+1)^2(2t+1)^t\cdot |S| \cdot 2^{(K-1-i)d}}{2^{(K-i)d}}\\ 
    &\leq \frac{K \cdot (t+1)^2(2t+1)^t\cdot |S|}{2^d} \leq \frac{\alpha}{4},
\end{align*}
where the last inequality holds for large enough $d$. We obtain that
\[
\Pr[W \cap \mathcal{G}] = \E[A_{K}] \geq \frac{3\alpha}{4} - \frac{\alpha}{4} = \frac{\alpha}{2}\,. \qedhere \]
\end{proof}

\subsection{Proof of \Lem{quadratic_t}} \label{sec:quadratic_witness}

In this section, we prove \Lem{quadratic_t} on the fraction of large violation structures for which all triples $(x_i^{(\ell)}, y_{j}^{(\ell)}, z)$, where  $(i,j,\ell) \in [I]\times [J]\times [t+1]$, %
violate quadraticity. Our proof builds on a result of \cite{AlonKKLR05}.  

\begin{proof}[Proof of \Lem{quadratic_t}]
Let $\eta$ denote the fraction of violating triples for $f$, i.e.,
\begin{equation*}
   \eta := \Pr_{x,y,z \sim \{0,1\}^d}[T_f(x,y,z) = 1].
\end{equation*}
The distance of $f$ to quadraticity, denoted by $\eps_f$, is the minimum of $\Pr_x[f(x)\neq g(x)]$ over all quadratic functions $g$ over the same domain as $f$.
Using this notation, we state a result from \cite{AlonKKLR05} for the special case of quadraticity.

\begin{claim}[Theorem 1 of \cite{AlonKKLR05}]\label{clm:bat}
For all $f$, we have $\eta \geq \min( \frac{7}{3}\eps_f, \frac{1}{40})$.
\end{claim}

 Let $\eta'$ denote the left-hand side of \Eqn{violation}, that is, the probability that for $x_i^{(\ell)}, y_j^{(\ell)}, z \sim \{0,1\}^d$,  all triples $(x_i^{(\ell)}, y_j^{(\ell)}, z)$ are violating, where  $(i,j,\ell) \in [I]\times [J]\times [t+1]$. \Clm{quadratic_cauchy} lower bounds $\eta'$ in terms of $\eta$ for all values of $\eps_f$. \Clm{small_eps} lower bounds $\eta'$ for small values of $\eps_f$. We  combine these results and use \Clm{bat} to conclude the proof of the lemma. 
    
\begin{claim} \label{clm:quadratic_cauchy}
For all $f$  and points $x_i^{(\ell)}, y_j^{(\ell)}, z \sim \{0,1\}^d$, where $i \in [I], j\in[J], \ell\in[t+1]$, it holds that
\begin{equation*}
  \Pr\Big[ \bigcap_{i \in [I], j\in[J], \ell\in[t+1]}[T_f(x_i^{(\ell)}, y_j^{(\ell)}, z) = 1] \Big] \geq  \eta^{IJ(t+1)}.
\end{equation*}
\end{claim}
\begin{proof}
    The proof uses the H\"{o}lder's inequality as its key ingredient.

    \begin{claim}[H\"{o}lder's inequality]\ Let\ $p, q \geq 1$\ such\ that\
      $\frac 1p+\frac 1q=1$.\ For\ all\ vectors\ $(a_1, \dots, a_n),$
 $(b_1, \dots, b_n) \in \R^{n}$,  %
    \begin{equation*}
        \sum_{i \in [n]}  |a_ib_i| \leq \Big( \sum_{i \in [n]}  |a_i|^{p} \Big)^{1/p} \Big( \sum_{i \in [n]}  |b_i|^{q} \Big)^{1/q}.
    \end{equation*}
    \end{claim}

    For $x_i^{(\ell)}, y_j^{(\ell)}, z \sim \{0,1\}^d$, let $T_{ij}^{(\ell)}$ be the event $T_f(x_i^{(\ell)}, y_j^{(\ell)}, z) = 1$. Then
    \begin{align*}
          \Pr\Big[ \bigcap_{i \in [I], j \in [J],  \ell \in [t+1]} T_{ij}^{(\ell)}\Big] &= \sum_{u \in \{0,1\}^d}     \Pr\Big[  \bigcap_{i \in [I], j \in [J],  \ell \in [t+1]} T_{ij}^{(\ell)} \:|\:  z = u\Big]\Pr[ z = u] \\
          &= \frac{1}{2^{d}}\sum_{u \in \{0,1\}^d}     \Pr\Big[  \bigcap_{i \in [I], j \in [J],  \ell \in [t+1]} T_{ij}^{(\ell)} \:|\: z=u\Big] \\
          &= \frac{1}{2^{d}}\sum_{u \in \{0,1\}^d}     \Pr\Big[ \bigcap_{i \in [I], j \in [J]} T_{ij}^{(1)} \:|\:z = u\Big]^{t+1} \\
          &\geq \Big( \frac{1}{2^{d}} \sum_{u \in \{0,1\}^d}     \Pr\Big[ \bigcap_{i \in [I], j \in [J]} T_{ij}^{(1)} \:|\:z = u\Big] \Big)^{t+1} \\
          &= \Pr\biggl[ \bigcap_{ i \in [I], j \in [J]} T_{ij}^{(1)}\biggr]^{t+1},
    \end{align*}
    where the first equality holds by the law of total probability; the third equality holds because, conditioned on $z$ taking a specific value,  the events $\bigcap_{i \in [I],j\in[J]}T_{ij}^{(\ell)}$ for $\ell\in[t+1]$ are independent and have the same probability; the inequality follows from the H\"{o}lder's inequality. 
    
    We use a similar argument to obtain
    \begin{equation*}
        \Pr\Big[ \bigcap_{ i \in [I], j \in [J]} T_{ij}^{(1)}\Big] \geq \Big(\Pr\Big[ \bigcap_{ j \in [J]} T_{1j}^{(1)}\Big]\Big)^{I},
    \end{equation*}
    where we condition on the values of $y_1^{(1)}, \dots, y_J^{(1)}$, and $z$. Similarly, by conditioning on the values of $x_1^{(1)}$ and $z$, we obtain
     \begin{equation*}
        \Pr\Big[ \bigcap_{ j \in [J]} T_{1j}^{(1)}\Big] \geq (\Pr[T_{11}^{(1)}])^{J}.
    \end{equation*}
    Since $\Pr[T_{11}^{(1)}] = \eta$, the claim follows. 
\end{proof}

Next, we consider the case when $\eps_f$ is small. For $u_1, u_2, u_3 \in \{0,1\}^d$, let $\mathrm{span}(u_1, u_2, u_3)$ be the set of points $\bigoplus_{i \in T} u_i$ for $ \emptyset \neq T \subseteq [3]$ and 
 \begin{align*}
      S = \bigcup_{i \in [I], j\in [J], \ell \in [t+1]} \mathrm{span}(x_i^{(\ell)}, y_j^{(\ell)}, z) .
 \end{align*}
 The set $S$ has $(I+J)(t+1)+1$ singletons, at most $I(J+1)(t+1)$ double sums, and at most $IJ(t+1)$ triple sums. Therefore, $|S| \leq 3IJ(t+1)$.  
 
 \begin{claim} \label{clm:small_eps}
 Suppose $\eps_f \leq \frac{1}{2|S|}$. Then $\eta' \geq \frac{\eps_f}{2}$.
 \end{claim}
 \begin{proof}
 Let $g$ be a closest quadratic function to $f$. Any two elements of $S$ are uniformly and independently distributed in $\{0,1\}^d$. Then, for $x_i^{(\ell)}, y_j^{(\ell)}, z \sim \{0,1\}^d$, we have
 \begin{align*}
     \eta' &\geq \Pr[f(z) \neq g(z) \text{ and } f(u) = g(u) \:\forall\: u \in S\setminus \{z\}] \\
     &\geq \Pr[f(z) \neq g(z)] - \sum_{u \in S\setminus\{z\}}[f(z) \neq g(z) \text{ and } f(u) \neq g(u)] \\
     &\geq \eps_f - (|S|-1)\eps_f^2 = \eps_f(1 - (|S|-1)\eps_f).
 \end{align*}
 If $\eps_f \leq \frac{1}{2|S|}$, then $1 - (|S|-1)\eps_f \geq 1 - |S| \cdot \frac{1}{2|S|} \geq \frac{1}{2}$, which concludes the proof. 
 \end{proof}
 
 Suppose $\frac{1}{2|S|} \leq \eps_f \leq \frac{3}{7\cdot 40}$. In this case, by~\Clm{bat}~and~\Clm{quadratic_cauchy} and using the fact that $|S| \leq 3IJ(t+1)$, we have 
 \begin{align*}
 \eta' \geq \Big(\frac{7}{3}\eps_f\Big)^{IJ(t+1)} \geq \Big(\frac{7}{6\cdot |S|}\Big)^{IJ(t+1)} \geq  \frac{7}{(18  \cdot IJ(t+1))^{IJ(t+1)}}.
 \end{align*}
 Finally, if $\eps_f \geq \frac{1}{40}$, then again by~\Clm{bat}~and~\Clm{quadratic_cauchy}, we have $\eta' \geq \frac{1}{40^{(IJ(t+1))}}$. We have obtained that $\eta' \geq \min \bigl( \frac{\eps}{2}, \frac{7}{(18 \cdot IJ(t+1))^{I J (t+1)}} \bigr)$. 
\end{proof}

\subsection{Proof of \Lem{erasures}} \label{sec:lem_erasures}

In this section, we prove \Lem{erasures}, which shows that conditioned on the good event $W \cap \mathcal{G}$, all queries in one round of \Alg{quadratic_t} are successfully obtained. In particular, this probability only depends on the per-query erasure budget $t$. 

\begin{proof}[Proof of \Lem{erasures}]
    We first prove a lower bound on the probability that the queries made in Step~\ref{step:query_xy} are successfully obtained.
\begin{claim}\label{clm:connect}
Conditioned on the event $W\cap \mathcal{G}$, the probability that in one execution of Step~\ref{step:query_xy} at level $m \in [0,t]$, all queries in the step are successfully obtained is at least
\begin{equation*}
    \Big(\frac{1}{(t^2+t)(2t+1)^{t-m}+1} \Big)^{(t+1)(2t+1)^{t-m}}.
\end{equation*}
\end{claim}
\begin{proof}
Fix the values of $\ell,m, j_1, \dots, j_{m}$ for the given %
execution
of Step~\ref{step:query_xy} and let $\bj = (\ell,j_1, \dots, j_{m})$. 
We can assume, without loss of generality, that Step~\ref{step:query_x} considers the tuples $(j_1, \dots, j_m)$ in the lexicographic order.
When $j_m = 1$, then right before Step~\ref{step:query_xy} is executed, we have $|S_{\bj^{(-1)}}| = (t+1)(2t+1)^{t-m+1}$. The size of the set $S_{\bj^{(-1)}}$ decreases as the value of $j_m$ increases, so that for $j_m = t+1$, we have  
\begin{align}
    |S_{\bj^{(-1)}}| &= (t+1)(2t+1)^{t-m+1} - t(t+1)(2t+1)^{t-m} \nonumber \\
    &= (t+1)(2t+1)^{t-m}(2t+1-t) = (t+1)^2 (2t+1)^{t-m}.\nonumber
\end{align}
Thus, right before Step~\ref{step:query_xy} is executed, the size of  $S_{\bj^{(-1)}}$ is at least $(t+1)$ times the size of the subset $S_{\bj}$.  

Let $s$ denote the size of $S_{\bj^{(-1)}}$ right before the execution of Step~\ref{step:random_subset}. In  Step~\ref{step:random_subset} we sample a uniformly random subset $S_{\bj}$ of $S_{\bj^{(-1)}}$ of size $s' = (t+1)(2t+1)^{t-m}$ and in Step~\ref{step:query_xy} we query $x \oplus y_{\bj}$ for all $x \in S_{\bj}$. Conditioned on the event $W \cap \mathcal{G}$, all sums $x \oplus  y_{\bj}$, where $x\in S_{\bj}$, are distinct and nonerased right before $y_{\bj}$ is queried. On the $i$-th query of Step~\ref{step:query_xy}, the tester selects a uniformly random $x$ out of $s-(i-1)$ elements. Right before the $i$-th query of the tester, the oracle can have erased at most $ti$ of the sums $x \oplus y_{\bj}$, for $x \in S_{\bj^{(-1)}}$, since the adversary is not aware of the points belonging to $S_{\bj}$ until their sums with $ y_{\bj}$ are queried in Step~\ref{step:query_xy}. Therefore, the probability that the tester successfully obtains the sum on its $i$-th query is at least $1 - \frac{ti}{s-i+1}$, where $i \in [s']$. We argued that, for all values of $j_m$, right before the execution of Step~\ref{step:query_xy}, it always holds that $s \geq (t+1)s'$. Therefore, the probability that the tester successfully obtains a sum is always positive. In particular, using the bound $s \geq (t+1)s'$, the probability that all queries in Step~\ref{step:query_xy} are successfully obtained is at least
\begin{equation*}
        \prod_{i=1}^{s'} \Big( 1 - \frac{ti}{s-i + 1} \Big) 
        \geq \Big( 1 - \frac{ts'}{s-s' + 1} \Big)^{s'} 
        \geq \Big( \frac{1}{(t+1)s' - s' + 1} \Big)^{s'}
        =  \Big(\frac{1}{s't+1} \Big)^{s'}.
    \end{equation*}
Substituting $s' = (t+1)(2t+1)^{t-m}$ concludes the proof. 
\end{proof}
 
In Steps~\ref{step:query_y}-\ref{step:query_z}, the tester makes only singleton and double queries. Conditioned on $W \cap \mathcal{G}$, all singleton queries are successfully obtained. All double queries are made in Step~\ref{step:query_xy}. By \Clm{connect}, the probability that all double queries in Step~\ref{step:query_xy} are successfully obtained depends only on $t$. 
    
    Before $z$ is queried in Step~\ref{step:query_z}, conditioned on the event $W \cap \mathcal{G}$, all double and triple sums of the form $y_{\bj^{(m)}}  \oplus z$, $x  \oplus z$, and $x  \oplus y_{\bj^{(m)}} \oplus  z$, where $m \in [0,t]$, $\bj = (\ell, j_1 \dots, j_t) \in [t+1]^{(t+1)}$, and $x \in S_{\bj}$, are distinct and nonerased. 
    After $z$ is queried, the oracle can perform at most $t$ erasures. Therefore, there exists $\ell \in [t+1]$, say $\ell=1$, such that all double and triple sums $y_{\bj^{(m)}}  \oplus z$, $x  \oplus z$, and $x  \oplus y_{\bj^{(m)}} \oplus  z$, where $m \in [0,t]$, $\bj = (1, j_1 \dots, j_t) \in [t+1]^{(t+1)}$, and $x \in S_{\bj}$, are nonerased after $z$ is queried. With probability $\frac{1}{t+1}$, the tester samples $\ell=1$ in Step~\ref{step:query_indices}. 
    By a similar reasoning, with probability $\frac{1}{(t+1)^{t}}$ the tester samples values of $j_1, \dots, j_{t}$ in Step~\ref{step:query_indices}, say $j_1 = \dots = j_t = 1$, such that all queries $y_{\bj^{(m)}} \oplus z$, where $m \in [0,t]$, are successfully obtained. In addition, right before  $y_{\bj^{(t)}} \oplus z$ is queried, all sums of the form $x \oplus z$ and $x\oplus y_{\bj^{(m)}} \oplus z$, where $\bj = (1, \dots, 1, j_t) \in [t+1]^{(t+1)}$, $m \in [0,t]$, and $x \in S_{\bj}$, are nonerased.

    After  $y_{\bj^{(t)}} \oplus z$ is queried, the oracle performs at most $t$ erasures, so there exists $x \in S_{\bj}$ such that $x \oplus z$ and $x\oplus y_{\bj^{(m)}} \oplus z$, where $\bj = (1, 1, \dots, 1)$ and $m \in [0,t]$, are nonerased. 
    With probability $\frac{1}{t+1}$, the tester samples this $x$ in Step~\ref{step:query_zy}, and with probability $\frac{1}{t+1}$, it samples $m \in [0,t]$ such that the triple sum $x\oplus y_{\bj^{(m)}} \oplus z$ is successfully obtained in Step~\ref{step:3sum}. Thus, with probability $\frac{1}{(t+1)^{t+3}}$, all queries in Steps~\ref{step:query_indices}-\ref{step:3sum} are successfully obtained. 
    Therefore the probability that all queries in one round of \Alg{quadratic_t} are successfully obtained, conditioned on $W \cap \mathcal{G}$, is $1/c_t$, where $c_t$ is a constant depending only on $t$. 
\end{proof}

\section{Online-erasure-resilient sortedness testing}\label{sec:sortedness}

In this section, we prove \Thm{sortedness_real} on the impossibility of general online-erasure-resilient sortedness testing and \Thm{sortedness_uniform} on online-erasure-resilient sortedness testing of sequences with few distinct values. 

\subsection{Impossibility of online-erasure-resilient sortedness testing}

In this section, we prove \Thm{sortedness_real} which shows that online-erasure-resilient testing of sortedness of integer sequences is impossible.

\begin{proof}[Proof of \Thm{sortedness_real}]
	By Yao's principle (see \Cor{yao_2}), it is enough to give a pair of distributions, one over monotone functions over $[n]$ and the other one over functions over $[n]$ that are far from monotone, and an adversarial strategy, such that there is no deterministic algorithm that can distinguish the distributions with high probability, when given access to them via a $1$-online-erasure oracle that uses the given strategy. 
	
	\noindent Let $n \in \mathbb{N}$ be even. Consider the following distributions on functions $f:[n] \to [n]$. 
	\medskip
	
	\noindent\textbf{$\calD^+$ distribution:} Independently for all $i \in [n/2]$:  %
	\begin{itemize}
		\item %
		$f(2i-1) \gets 2i-1$ and $f(2i) \gets 2i -1$, with probability $1/3$.
		\item $f(2i-1) \gets 2i-1$ and $f(2i) \gets 2i$, with probability $1/3$.
		\item $f(2i-1) \gets 2i$ and $f(2i) \gets 2i$, with probability $1/3$. %
	\end{itemize}
	
	\noindent\textbf{$\calD^-$ distribution:} Independently for all $i \in [n/2]$: %
	\begin{itemize}
		\item $f(2i-1) \gets 2i$ and $f(2i) \gets 2i -1$, with probability $1/3$. %
		\item $f(2i-1) \gets 2i-1$ and $f(2i) \gets 2i$, with probability $2/3$. %
	\end{itemize}
	
	Every function sampled from the distribution  $\calD^+$ is monotone. 
	
	For a function sampled from the distribution $\calD^-$, the expected number of index pairs $(2i-1, 2i)$ such that the respective function values are $2i$ and $2i-1$ is equal to $n/6$. 
	By a Chernoff bound, the probability that the number of such index-pairs is less than $n/12$ is at most $\exp(-n/48)$.
	In other words, with probability at least $1 - \exp(-n/48)$, a function sampled from the $\calD^-$ distribution is $1/12$-far from being monotone.
	
	Consider a deterministic adaptive algorithm $T$ with two-sided error for online-erasure-resilient $\frac{1}{12}$-testing of monotonicity of functions from $[n]$ to $[n]$. 
	Assume that $T$ is given access to a function sampled from $\calD^+$ or $\calD^-$ with equal probability, where the access is via a $1$-online-erasures oracle.
	The oracle uses the following strategy.
	For each $i \in [n/2]$, if $T$ queries the point $2i-1$, then the oracle erases the function value at the point $2i$; if $T$ queries the point $2i$ then the oracle erases the function value at the point $2i-1$. The oracle erases exactly one function value between queries. Given this strategy, the tester $T$ never sees a violation to monotonicity. 
	
	The distribution of function values restricted to any particular index is identical for distributions $\calD^+$ and $\calD^-$. Hence, the tester $T$ cannot distinguish the distributions, no matter how many queries it makes. 
\end{proof}

\subsection{Online-erasure-resilient testing of sortedness with small $r$}

In this section, we prove \Thm{sortedness_uniform} by showing that $O(\frac{\sqrt{r}}{\eps})$ uniform queries are sufficient for $t$-online-erasure-resilient $\eps$-testing of sortedness, when $r$, the number of distinct values in the sequence, is small. 
Pallavoor et al.~\cite{PRV18} gave a 1-sided error $\eps$-tester for sortedness of real-valued sequences containing at most $r$ distinct values that makes $O(\frac{\sqrt{r}}{\eps})$ uniform and independent queries. 
The proof of \cite[Theorem 1.5]{PRV18} is the starting point for our analysis. 

\begin{proof}[Proof of \Thm{sortedness_uniform}] 
We call $(u, v) \in [n]^2$ a \emph{violating} pair if $u < v$ and $f(u) > f(v)$. The tester, given input $\eps \in (0,1)$ and oracle access to a function $f$ with image size at most $r$, queries $f$ on $\frac{2c\sqrt{r}}{\eps}$ points selected uniformly and independently and rejects if and only if it queries $u, v \in [n]$ such that $(u,v)$ is a violating pair.
Since the decision made by the tester does not depend on the specific values being queried, but only on the ordering of the queried values, we can assume without loss of generality that the input function is of the form $f:[n] \to [r]$. Clearly, the tester accepts all sorted functions. We show that for $c = 32$, if $f$ is $\eps$-far from sorted, the tester rejects with probability at least $\frac{2}{3}$. 

Suppose $f$ is $\eps$-far from sorted. Define the violation graph $G([n], E)$ to be the graph whose edges are the violating pairs $(u, v) \in [n]^2$. Dodis et al.~\cite[Lemma 7]{DGLRRS99} show that $G$ has a matching $M$  of size at least $\eps n /2$. For each range value $i \in [r]$, let $ M_i$ be the set of edges  $(u, v) \in M$ whose lower endpoint $u$ has value $i$, i.e., $f(u) = i$. Let $L_i$ consist of the $|M_i|/2$ smallest indices amongst the lower endpoints of the edges in $M_i$. Let $U_i$ consist of the $|M_i|/2$ largest indices amongst the upper endpoints of the  edges in $M_i$. It is not hard to see that, for every $i\in[r]$, all pairs in $L_i \times U_i$ are violating. By construction, 
\begin{equation}\label{eq:matching}
    \sum_{i \in [r]} |L_i| = \sum_{i \in [r]} |U_i| = \sum_{i \in [r]} |M_i|/2 = \frac{|M|}{2} \geq \frac{\eps n }{4}. 
\end{equation}
We show that with probability at least $\frac{2}{3}$, for some $i \in [r]$, the tester samples a nonerased element both from $L_i$ and from $U_i$, in which case it rejects.

Our proof relies on the following claim from \cite{GGLRS00}, for which we first introduce some notation. Let $L_1,\dots, L_r$ be disjoint subsets of $[n]$. Let $p_i = \frac{|L_i|}{n}$ and $\rho = \sum_{i \in [r]} p_i$. Suppose we sample $c\sqrt{r}/ \rho$ uniform points from $[n]$. Let $I$ be a random variable denoting the set of indices $i \in [r]$ such that the sample contains at least one element from the set $L_i$. 

\begin{claim}[Claim 1 of \cite{GGLRS00}] \label{clm:gglrs}
With probability $1 - \eee^{-c}$ over the choice of the sample, 
\begin{equation*}
    \sum_{i \in I} p_i \geq \frac{\rho}{\sqrt{r}}.
\end{equation*}
\end{claim}

 Split the sample into two parts of size $\frac{c\sqrt{r}}{\eps}$ each. Let $L_i' \subseteq L_i$ be the set of nonerased indices of $L_i$ after the tester queries the first part of the sample, and define $U_i'$ similarly. Let $p_i' = \frac{|L_i' |}{n}$, $q_i' = \frac{|U_i' |}{n}$, and $\rho' = \sum_{i \in [r]} \min(p_i', q_i')$. Without loss of generality assume $p_i' \leq q_i'
$ for all $i \in [r]$. The adversary performs at most $\frac{c\sqrt{r}t}{\eps}$ erasures after the first part of the sample is queried, therefore 
\begin{equation*}
    \rho' = \sum_{i \in [r]} \frac{|L_i' |}{n} \geq  \Big(\sum_{i \in [r]} \frac{|L_i |}{n}\Big) - \frac{c\sqrt{r}t}{\eps n} \geq \frac{\eps}{4} - \frac{c\sqrt{r}t}{\eps n} \geq \frac{\eps}{8},
\end{equation*}
where we use \Eqn{matching} and the assumption that $\sqrt{r} \leq r \leq \frac{\eps^2 n}{8c^2 t}$.
Let $I'$ be a random variable denoting the set of indices $i \in [r]$ such that at least one element of $L_i'$ is contained in the first part of the sample. 
By \Clm{gglrs}, with probability $1- \eee^{-\frac{c}{8}}$ over the first part of the sample,
\begin{equation}\label{eq:pi}
     \sum_{i \in I'} p_i' \geq \frac{\eps}{8\sqrt{r}}.
\end{equation}
The probability that the second part of the sample does not include a nonerased element from $\cup_{i \in [r]} U_i'$ is
\begin{equation*}
    \Big ( 1 - \sum_{i\in [r]} q_i' + \frac{c\sqrt{r}t}{\eps n} \Big )^\frac{c\sqrt{r}}{\eps} \leq \Big ( 1 -  \frac{\eps}{8\sqrt{r}} + \frac{\eps}{16\sqrt{r}} \Big )^\frac{c\sqrt{r}}{\eps} \leq \eee^{-\frac{c}{16}},
\end{equation*}
where we use \Eqn{pi} to lower bound $\sum_{i\in [r]} q_i'$, and the assumption that $r < \frac{\eps^2 n}{ct}$ to upper bound the number of erasures. 
Therefore, the probability that for some $i \in [r]$ the tester samples a nonerased element both from $L_i$ and $U_i$ and rejects is at least
\begin{equation*}
 (1 - \eee^{-\frac{c}{8}}) \cdot (1 - \eee^{-\frac{c}{16}})\geq \frac{2}{3},
\end{equation*}
where the inequality holds for  $c= 32$. 
\end{proof}

\section{Impossibility of online-erasure-resilient Lipschitz testing}\label{sec:lipschitz}
In this section, we prove \Thm{lipschitz}.
We start by showing that $1$-online-erasure-resilient Lipschitz testing  of functions $f:[n] \to \{0,1,2\}$ is impossible. 
\begin{proof}[Proof of \Thm{lipschitz}, part 1]
We prove the theorem using Yao's principle. The proof is analogous to that of \Thm{sortedness_real}
once the hard distributions are described. In particular, the strategy of the oracle is identical to that used in the proof of \Thm{sortedness_real}.

Let $n \in \mathbb{N}$ be a multiple of $4$. The following distributions are over functions $f:[n] \to \{0,1,2\}$. \\

\noindent\textbf{Distribution $\mathcal{D}^+$:} Independently, for all $i \in[n]$ such that $i \bmod 4 \equiv 1$:
\begin{itemize}
    \item $f(i) \gets 0$ and $f(i+1) \gets 1$, with probability $1/2$.
    \item $f(i) \gets 1$ and $f(i+1) \gets 2$, with probability $1/2$.
\end{itemize}
Independently, for all $i \in[n]$ such that $i \bmod 4 \equiv 3$:
\begin{itemize}
    \item $f(i) \gets 1$ and $f(i+1) \gets 0$, with probability $1/2$.
    \item $f(i) \gets 2$ and $f(i+1) \gets 1$, with probability $1/2$.
\end{itemize}

\noindent\textbf{Distribution $\mathcal{D}^-$:} Independently, for all $i$ such that $i \bmod 4 \equiv 1$:
\begin{itemize}
    \item $f(i) \gets 0$ and $f(i+1) \gets 2$, with probability $1/2$.
    \item $f(i) \gets 1$ and $f(i+1) \gets 1$, with probability $1/2$.
\end{itemize}
Independently, for all $i$ such that $i \bmod 4 \equiv 3$:
\begin{itemize}
    \item $f(i) \gets 2$ and $f(i+1) \gets 0$, with probability $1/2$.
    \item $f(i) \gets 1$ and $f(i+1) \gets 1$, with probability $1/2$.
\end{itemize}

Every function sampled from the distribution  $\calD^+$ is Lipschitz. For a function sampled from the distribution $\calD^-$, the expected distance to being Lipschitz is $1/4$.
The rest of the proof is similar to that of \Thm{sortedness_real}.
\end{proof}

Next, we prove the second part of \Thm{lipschitz} showing that $1$-online-erasure-resilient Lipschitz testing of functions $f:\{0,1\}^d \to \{0,1,2\}$ is impossible. 

\begin{proof}[Proof of \Thm{lipschitz}, part 2] The proof is analogous to the proof of the first part of the theorem once the hard distributions are described.

The following distributions are over functions $f:\{0,1\}^d \to \{0,1,2\}$. Let $e_1 \in \{0,1\}^d$ denote the first standard basis vector.

\subparagraph*{Distribution $\mathcal{D}^+$:} Independently, for all $x \in \{0,1\}^d$ such that $x_1 = 0$:
\begin{itemize}
    \item $f(x) \gets 0$ and $f(x \oplus e_1) \gets 1$, with probability $1/2$.
    \item $f(x) \gets 1$ and $f(x \oplus e_1) \gets 2$, with probability $1/2$.
\end{itemize}

\subparagraph*{Distribution $\mathcal{D}^-$:} Independently, for all $x \in \{0,1\}^d$ such that $x_1 = 0$:
\begin{itemize}
    \item $f(x) \gets 0$ and $f(x \oplus e_1) \gets 2$, with probability $1/2$.
    \item $f(x) \gets 1$ and $f(x \oplus e_1) \gets 1$, with probability $1/2$.
\end{itemize}
Every function sampled from the distribution  $\calD^+$ is Lipschitz. For a function sampled from the distribution $\calD^-$, the expected distance to being Lipschitz is $1/4$.
The rest of the proof is similar to that of \Thm{sortedness_real}. This completes the proof of \Thm{lipschitz}.
\end{proof}

\section{Computation in the presence of online corruptions} \label{sec:corruptions}

This section discusses implications of our  results to computation in the presence of online corruptions.

\subsection{Online-corruption-resilience from online-erasure-resilience in testing}\label{sec:two_sided}

In this section, we prove \Lem{no_erasures}, showing that an algorithm  that accesses its input via an online-erasure oracle and, with high probability, encounters no erasures and outputs a correct answer, is also correct with high probability for all adversarial strategies (and for the same computational task) when it accesses its input via an online-corruption oracle.

\begin{proof}[Proof of \Lem{no_erasures}]
 Fix an input function $f$ and a strategy of the $t$-online-corruption oracle for this function. Let $\cO^{(c)}$ denote an oracle that uses this strategy. Map the strategy of $\cO^{(c)}$ to a strategy for a $t$-online-erasure oracle $\cO^{(e)}$, so that whenever $\cO^{(c)}$ modifies the value at a point $x$ to $\cO^{(c)}(x) \neq f(x)$, the online-erasure oracle erases the value at the same point, i.e.,  it sets $\cO^{(e)}(x) = \perp$. 
 
 Fix the random coins of an execution of $T$, where access to $f$ is via $\cO^{(e)}$, and for which all queries are successfully obtained and $T$ outputs a correct answer. Let $T^{(e)}$ be the execution of $T$ with those coins. Similarly, let $T^{(c)}$ be the execution of $T$ with the same coins, when access to $f$ is via $\cO^{(c)}$. Since $T^{(e)}$ encounters no erasures, $T^{(c)}$ will make the same queries as $T^{(e)}$ and get the same query answers. Consequently, $T^{(c)}$ will also output the same (correct) answer as  $T^{(e)}$. Therefore, when $T$ accesses $f$ via an online-corruption oracle, it outputs a correct answer with probability at least $2/3$.
\end{proof}

\subsection{Online-corruption-resilient linearity tester}\label{sec:linearity_corrupt}

In this section, we show the existence of a 2-sided error, online-corruption-resilient tester for linearity.

\begin{corollary}\label{cor:linearity_corruptions}
There exists 
$c_0 \in (0,1)$ and a nonadaptive,~$t$-online-corruption-resilient $\eps$-tester for linearity of functions~$f\colon \{0,1\}^d \to \{0,1\}$ that  works for $t \leq c_0\cdot\eps^{5/4} \cdot  2^{d/4}$ and makes $O(\frac{1}{\eps}\log \frac{t}{\eps})$ queries.
\end{corollary}

\begin{proof} %

Consider \Alg{linearity_3}, modified to use $q=2\log \frac{Ct}{\eps^2}$ for $C=3000$. Note that the larger the size $q$ of the reserve, the higher is the probability that \Alg{linearity_3} detects a violation to linearity. For the choice of  $q=2\log \frac{3000t}{\eps^2}$, \Alg{linearity_3} is also a 1-sided error online-erasure-resilient tester. Furthermore, for this choice of $q$,  \Alg{linearity_3} rejects with probability at least $5/6$ (as opposed to $2/3$) when $f$ is $\eps$-far from linear. 

By \Lem{no_erasures}, it remains to show that, for all adversarial strategies, \Alg{linearity_3} queries an erased point during its execution with probability at most $1/6$.
Fix an adversarial strategy and one round of \Alg{linearity_3}. (Recall that each iteration of the outer \textbf{repeat} loop in Steps~\ref{step:repeat}-\ref{step:reject} is called a {\em round}). Let $G$ be the event defined in \Lem{distinct_and_nonerased} for the fixed round. Then $\Pr[\overline{G}] \leq \frac{\eps}{800}$. 

Let $B$ denote the event that the algorithm queries an erased point during the fixed round. Then
\begin{align*}
    \Pr[B] = \Pr[B \mid G]\Pr[G] + \Pr[B \mid \overline{G}]\Pr[\overline{G}] \leq \Pr[B \mid G] + \Pr[\overline{G}].
\end{align*}
To upper bound $\Pr[B \mid G]$ it suffices to upper bound the probability that the algorithm queries an erased sum in some iteration of Step~\ref{step:sample-I} of the fixed round. Since $G$ occurred, there are at least $2^{q-1}-1$ distinct sums that can be queried in Step~\ref{step:sample-I}, all of them nonerased at the beginning of the round. \Alg{linearity_3} makes at most $q + 4\cdot 2^j \leq q + \frac{32}{\eps}$ queries in this round. Thus, the fraction of these sums erased during the round is at most
\begin{align*}
    t\cdot\frac{q + \frac{32}{\eps}}{2^{q-1}-1}
   &\leq  t\cdot \Big(\frac{1}{2^{q/2}} + \frac{70}{\eps \cdot 2^{q}} \Big)
   \leq 
 \frac{\eps^2}{C} + \frac{70\eps^3}{4t \cdot C^2 }  
   \leq   \eps^2 \Bigl( \frac{1}{C} +  \frac{8.75}{C^2} \Bigr)
   \leq \frac{\eps^2}{88\cdot 32},
\end{align*}
where in the first inequality we used that $\frac{q}{2^{q-1}-1} \leq \frac{1}{2^{q/2}}$  for $q \geq 9$ and that $(2^{q-1}-1)/32\geq  2^{q}/70$ for $q\geq 5$ (note that $q \geq 2\log(3000t/\eps^2)\geq  25$), in the second inequality we used $q \geq 2\log(3000t/\eps^2)$, and in the third inequality we used $\eps \leq \frac{1}{2}$. 

Since there are at most $4 \cdot 2^j \leq \frac{32}{\eps}$ iterations of Step~\ref{step:sample-I}, we obtain by a union bound that  
\begin{align*}
    \Pr[B \mid G] \leq \frac{32}{\eps} \cdot \frac{\eps^2}{88\cdot 32} = \frac{\eps}{88}.
\end{align*}
Consequently, $\Pr[B] \leq \frac{\eps}{88} + \frac{\eps}{800} \leq \frac{\eps}{80}$. The number of rounds of \Alg{linearity_3} is at most
\begin{align*}
    \sum_{j=1}^{\log 8/\eps} \frac{8\ln 5}{2^j \eps } = \frac{8\ln 5}{\eps} \sum_{j=1}^{\log \frac{8}{\eps}}\frac{1}{2^j} \leq \frac{8 \ln 5}{\eps}. 
\end{align*}
Therefore, \Alg{linearity_3} queries an erased point during its execution with probability at most $\frac{8\ln 5}{\eps} \cdot \Pr[B] \leq \frac{1}{6}$. This concludes the proof.  
\end{proof}

\subsection{Online-corruption-resilient tester of sortedness with few distinct values}\label{sec:sortedness_corrupt}

In this section, we show  the existence of nonadaptive online-corruption-resilient tester for sortedness of sequences with few distinct values. 

\begin{corollary} \label{cor:sortedness_corrupt}
There exists a nonadaptive, $t$-online-corruption-resilient $\eps$-tester for sortedness of $n$-element sequences with at most $r$ distinct values. The tester makes $O(\frac{\sqrt{r}}{\eps})$ uniform and independent queries and works when $r<\frac{\eps^2 n}{c_0 t}$, where $c_0 > 0$ is a specific constant.
\end{corollary}

\begin{proof} %

Recall that the online-erasure-resilient sortedness tester makes $\frac{2c\sqrt{r}}{\eps}$ uniform and independent queries to $f \colon [n]\to [r]$. The tester always accepts when $f$ is sorted. We can easily modify the tester, without affecting its asymptotic query complexity, to reject with probability at least $5/6$ (as opposed to $2/3$) when $f$ is $\eps$-far from sorted. Then, by \Lem{no_erasures}, it suffices to show that, for all adversarial strategies, this tester queries an erased point with probability at most $1/6$. The adversary performs at most $\frac{2c\sqrt{r}t}{\eps}$ erasures during the execution of the algorithm. Since the tester performs its queries uniformly at random, the probability that a given query is erased is at most $\frac{2c\sqrt{r}t}{\eps n}$. By a union bound, the probability that some query of the tester is erased is at most $\frac{4c^2 r t}{\eps^2 n}$. Set $c_0=24c^2$. By our assumption that $r < \frac{\eps^2 n}{24c^2 t}$, we obtain that the probability that the tester samples an erased point is at most $\frac{1}{6}$. 
\end{proof}

\subsection{Finding witnesses in the presence of online corruptions}\label{sec:detect}

In this section, we prove \Lem{detect}. Specifically, we show how to modify a nonadaptive, 1-sided error, online-erasure-resilient tester for a property $\cP$ to get an algorithm that accesses the input via an online-corruption oracle and outputs a witness demonstrating the violation of $\cP$.

\begin{proof}[Proof of \Lem{detect}]
    Let $T$ be a nonadaptive, 1-sided error, $t$-online-erasure-resilient tester for $\mathcal{P}$. Let $f$ be $\eps$-far from $\cP$. Fix a strategy of the $t$-online-corruption oracle, and let $\cO^{(c)}$ denote the oracle running this strategy. Define a a $t$-online-erasure oracle $\cO^{(e)}$ that follows the same strategy as  $\cO^{(e)}$, except that $\cO^{(e)}$ erases the value at each point $x$ that $\cO^{(c)}$ modifies,
i.e., it sets $\cO^{(e)}(x) = \perp$. 
    
    If the tester $T$ accesses $f$ via $\cO^{(e)}$, it rejects with probability at least $2/3$.
    Fix the random coins of an execution of $T$ for which it rejects. 
    Let $T^{(e)}$ be the execution of $T$ with those coins (where access to $f$ is via $\cO^{(e)}$). Similarly, let $T^{(c)}$ be the execution of $T$ with the same coins when access to $f$ is via $\cO^{(c)}$. 
    Since $T$ is nonadaptive, $T^{(c)}$ and $T^{(e)}$ make the same queries. By definition of $\cO^{(e)}$, if $T^{(e)}$ obtains $O^{(e)}(x) = f(x)$ for a query $x$, then $T^{(c)}$ obtains $O^{(c)}(x) = f(x)$ for the same query. 
  
    Since $T^{(e)}$ is a 1-sided error tester, it can reject only if it successfully obtains values of $f$ on some queried points $x_1, \dots, x_k$ such that no $g \in \cP$ has $g(x_i) =f(x_i)$ for all $i \in [k]$. Then, for the same queries,  $T^{(c)}$ obtains $\cO^{(c)}(x_1) = f(x_1), \dots, \cO^{(c)}(x_k)=f(x_k)$.  Therefore, $T^{(c)}$ finds a witness of $f \notin \cP$. We modify $T$ to output a witness if it finds one. (Note that $T^{(c)}$ and $T^{(e)}$ might output different sets of points violating $\cP$, since $T^{(c)}$ will obtain additional values $\cO^{(c)}(x) \in \R$ for the queries $x$ for which $T^{(e)}$ obtained $\cO^{(e)}(x) = \perp$.) Now, if we run the modified $T$ with access to $f$ via $\cO^{(c)}$, it outputs a witness of $f \notin \cP$ with probability at least $2/3$. 
\end{proof}

\section{Yao's minimax principle for online erasure oracles}\label{sec:yao}
	
	We extend Yao's minimax principle \cite{Yao77} to the setting of online erasures, so that it can be used  to prove lower bounds on the query complexity of randomized online-erasure-resilient testers. In the standard property testing model, Yao's principle allows one to prove a lower bound on query complexity by describing a distribution on inputs that is hard to test, in the sense that a deterministic algorithm needs many queries to decide, with high probability over the input, whether an input sampled from the distribution has the property or is far from it. In the setting of online erasures, we show that one can prove a lower bound by describing a distribution on inputs as well as an oracle erasure strategy such that every deterministic algorithm that accesses the input selected from the distribution via the erasure oracle needs many queries to test a specified property with high  probability. 
	
	\begin{theorem}
		\label{thm:yao}
		To prove a lower bound $q$ on the worst-case query complexity of online-erasure-resilient randomized algorithms that perform a specified computational task with error probability strictly less than~$\frac{1}{3}$, it is enough to give a distribution $\cD$ on inputs of size $n$ and a (randomized) oracle erasure strategy $\mathcal{S}$, such that every deterministic algorithm making at most $q$ queries to an input drawn from $\cD$ via an online-erasure oracle that uses strategy $\mathcal{S}$, errs with probability at least $\frac{1}{3}$.
	\end{theorem}

	\begin{proof} %

		We prove a min-max inequality similar to Yao's original min-max inequality \cite{Yao77}. The inequality relates the error complexities of the algorithms, rather than their query complexity, which we will fix beforehand. 
		Consider a randomized algorithm for the specified task
  that, given access to an input object of size $n$
  via a $t$-online-erasure oracle, makes at most $q$ queries, and may be incorrect with some probability. We can represent the randomized algorithm as a probability distribution $\tau$ on the set of all (possibly incorrect) deterministic algorithms $A$ with query complexity at most $q$. Let $\mathcal{A}_q$ denote the set of such algorithms, and let $\tau(A)$ denote the probability of drawing algorithm $A$ from the set $\mathcal{A}_q$. 
		
		Denote by $X_n$ the set of inputs of size $n$ for the computational task.
		Let $\mathcal{S}_{q,n}$ be the set of deterministic erasure strategies of the $t$-online-erasure oracle for inputs of size $n$ against algorithms of query complexity at most $q$. An adversarial strategy can be represented by a decision tree that, based on the input and the queries of the algorithm, indicates which input entries to erase next. Since we are only considering inputs of fixed size and algorithms of bounded query complexity, the set of oracle strategies has finite cardinality.
		Oracle strategies can be randomized, so we consider a distribution on deterministic strategies. For ease of notation, we consider a joint distribution on inputs and oracle strategies, i.e., we  let $\sigma$ be a distribution on the input-strategy pairs $X_n \times \mathcal{S}_{q,n}$. Given an input $x \in X_n$, an adversarial strategy $s \in \mathcal{S}_{q,n}$, and an algorithm $A \in \cA_q$,
		let $I(A,x, s)$ be the indicator function for $A$ being incorrect on $x$, i.e., $I(A,x,s) = 1$ if $A$ is incorrect on $x$ against an oracle that uses strategy $s$, and $I(A,x,s) = 0$ otherwise. 
		
		We 
  prove
  the following min-max inequality
		\begin{equation}
		\min_\tau \max_{(x, s) \in X_n \times \mathcal{S}_{q,n}} \sum_{A \in \mathcal{A}_q} I(A,x,s)\tau(A)  \geq \max_\sigma \min_{A \in \mathcal{A}_q}\sum_{(x, s) \in X_n \times \mathcal{S}_{q,n}} I(A,x,s)\sigma(x,s).
		\label{eq:yao_error}
		\end{equation}
		
		We comment on how \Eqn{yao_error} implies \Thm{yao}. Suppose there exists some distribution $\sigma$ on input-strategy pairs, such that every $q$-query deterministic algorithm errs with probability at least $1/3$ on an input-strategy pair drawn from $\sigma$.
		This implies that the right-hand side of \Eqn{yao_error} is at least $\frac{1}{3}$. (Note that we can express $\sigma$ as a distribution $\mathcal{D}$ on the inputs and a randomized adversarial strategy $\mathcal{S}$).  By \Eqn{yao_error}, every $q$-query randomized algorithm errs with probability at least $\frac{1}{3}$ for every input and adversarial strategy. Therefore, an algorithm for the given computational task that errs with probability strictly less than $\frac{1}{3}$ has to make at least $q$ queries, which implies \Thm{yao}.

		We now prove \Eqn{yao_error}. Fix  
  distribution $\tau$ on algorithms and distribution $\sigma$ on input-strategy pairs. Then
		\begin{align}
		\max_{(x, s) \in X_n \times \mathcal{S}_{q,n}} \sum_{A \in \mathcal{A}_q} I(A,x,s)\tau(A) &= \biggl(\sum_{(x, s) \in X_n \times \mathcal{S}_{q,n}}  \sigma(x,s) \biggr) \max_{(x, s) \in X_n \times \mathcal{S}_{q,n}}  \sum_{A \in \mathcal{A}_q} I(A,x,s)\tau(A) \nonumber \\ 
		&\geq \sum_{(x, s) \in X_n \times \mathcal{S}_{q,n}}\sigma(x,s) \sum_{A \in \mathcal{A}_q} I(A,x,s)\tau(A)  \nonumber \\ 
		&= \sum_{A \in \mathcal{A}_q} \tau(A) \sum_{(x, s) \in X_n \times \mathcal{S}_{q,n}} I(A,x,s)\sigma(x,s) \label{eq:exchange} \\ 
		&\geq \sum_{A \in \mathcal{A}_q} \tau(A) \min_{A \in \mathcal{A}_q}\sum_{(x, s) \in X_n \times \mathcal{S}_{q,n}} I(A,x,s)\sigma(x,s)\nonumber  \\ 
		&= \min_{A \in \mathcal{A}_q}\sum_{(x, s) \in X_n \times \mathcal{S}_{q,n}} I(A,x,s)\sigma(x,s). \nonumber
		\end{align}
		Crucially, the fact that $\mathcal{A}_q, X_n$, and  $\mathcal{S}_{q,n}$ have finite cardinalities
		allows us to exchange the order of the sums in \Eqn{exchange}.
	\end{proof}

\subsection{A version of Yao's minimax with two distributions}\label{sec:yao_2}
In this section, we prove a corollary to \Thm{yao} which refers to the more common usage of Yao's minimax principle for proving lower bounds, where separate distributions on positive and negative instances are defined. The corollary is a generalization of  Claim~5 in \cite{RS06}.

	\begin{definition}The {\em statistical distance} between two discrete distributions $\mathcal{D}_1$ and $\mathcal{D}_2$ is
\begin{align*}
	    CD(\mathcal{D}_1, \mathcal{D}_2) = \max_{S \subseteq \textnormal{support}(\mathcal{D}_1) \cup \textnormal{support}(\mathcal{D}_2)} \Big( \Big| \Pr_{x \sim \mathcal{D}_1} [ x \in S] - \Pr_{x \sim \mathcal{D}_2} [ x \in S] \Big|\Big).
\end{align*}
	\end{definition}
	
\begin{definition} Given a $q$-query deterministic algorithm $\cA$, let $a(x)$ be the string of $q$ answers that $\cA$ receives when making queries to input object $x$. For a distribution $\cD$ and adversarial strategy $\cS$, let $\cD$-view be the distribution on strings $a(x)$ where the input object $x$ is sampled from $\cD$ and accessed via a $t$-online-erasure oracle using strategy~$\cS$. 
\end{definition}

		\begin{corollary}\label{cor:yao_2}
		
To prove a lower bound $q$ on the worst-case query complexity of online-erasure-resilient randomized algorithms for a promise decision problem,
		it is enough to give 
		\begin{itemize}
		\item a randomized adversarial strategy $\cS$,
			\item a distribution $\mathcal{D}^+$ on positive instances of size $n$, and
			
			\item a distribution $\mathcal{D}^-$ on instances of size $n$ that are negative with probability at least $\frac{6}{7}$,
		\end{itemize}
		such that $CD(\cD^+$-view$, \cD^-$-view$) < \frac{1}{6}$ for every deterministic $q$-query algorithm that accesses its input via an online-erasure oracle using strategy $\cS$. 
	\end{corollary} 

Our proof of \Cor{yao_2} is similar to
an argument in \cite{RS06}. We rely on the following claim from \cite{RS06}.

	\begin{claim}[Claim 4 of \cite{RS06}]\label{clm:bad_event}
	Let $E$ be an event that happens with probability at least $1-\delta$ under the distribution $\cD$ and let $\cB$ be the distribution $\cD$ conditioned on $E$. Then $CD(\cD, \cB) \leq \delta'$, where $\delta' = \frac{1}{1-\delta} - 1$. 
	\end{claim}

	\begin{proof}[Proof of \Cor{yao_2}]
	Let $A$ be a deterministic $q$-query algorithm for the promise decision problem. Given distributions $\mathcal{D}^+$ and $\mathcal{D}^-$, we define a distribution $\cD$ satisfying the conditions of \Thm{yao}. Let $E$ be the event that  $x \sim \mathcal{D}^-$ is a negative instance, and let $\cB$ be the distribution $\mathcal{D}^-$ conditioned on the event $E$. To obtain a sample from $\cD$, with probability $1/2$ draw a sample from $\cD^+$, and with probability $1/2$ draw a sample from $\cB$. We show that $A$ errs with probability at least $1/3$ when it accesses an input object $x\sim \cD$ via the online-erasure oracle that uses strategy $\cS$. 	Let $Q$ be the set of query-answers strings $a$ on which $A$ accepts.  By the law of total probability, 
	\begin{align*}
	    \Pr_{x\sim \cD}[A(x) \textnormal{ is correct}] &= \frac{1}{2} \Pr_{x\sim \cD^+}[A(x) \textnormal{ accepts}] + \frac{1}{2} \Pr_{x\sim \cB}[A(x) \textnormal{ rejects}] \\
	    &= \frac{1}{2} + \frac{1}{2}\Big( \Pr_{x\sim \cD^+}[A(x) \textnormal{ accepts}] -  \Pr_{x\sim \cB}[A(x) \textnormal{ accepts}]   \Big) \\
	    &= \frac{1}{2} + \frac{1}{2}\Big( \Pr_{a\sim \cD^+\text{-view}}[a \in Q] -  \Pr_{a\sim \cB\text{-view}}[a \in Q]  \Big) \\
	    &\leq \frac{1}{2} + \frac{1}{2}CD(\cD^+\text{-view}, \cB\text{-view}). 
	\end{align*}
	
	Note that $\cB\text{-view}$ is the same as the distribution $\cD^-\text{-view}$ conditioned on the event $E$. By \Clm{bad_event} and the assumption that $\Pr_{x\sim \cD^-}[E] \geq \frac{6}{7}$, we have $CD(\cB\text{-view}, \cD^-\text{-view}) \leq \frac{1}{6}$. By the triangle inequality, 
	\begin{equation*}
	    CD(\cD^+\text{-view}, \cB\text{-view}) \leq  
	    CD(\cD^+\text{-view}, \cD^-\text{-view}) +  
	    CD(\cB\text{-view}, \cD^-\text{-view},) < \frac{1}{6} + \frac{1}{6} = \frac{1}{3}.
	\end{equation*}
	Therefore, $\Pr_{x\sim \cD}[A(x) \textnormal{ is correct}] < \frac{1}{2} + \frac{1}{2} \cdot \frac{1}{3} = \frac{2}{3}$. 
	\end{proof}

\bibliographystyle{tocplain}   %

\bibliography{bibstrings,v019a001,bibtail}

\begin{tocauthors}
\begin{tocinfo}[kalemaj]
 Iden Kalemaj\\
 Graduate student\\
 Department of Computer Science\\
 Boston University\\
 Boston, MA, USA\\
 ikalemaj\tocat{}bu\tocdot{}edu \\   %
 \url{https://cs-people.bu.edu/ikalemaj/}      %
\end{tocinfo}
\begin{tocinfo}[raskhodnikova]
  Sofya Raskhodnikova\\
  Professor\\
  Department of Computer Science\\
  Boston University\\
  Boston, MA, USA\\
  sofya\tocat{}bu\tocdot{}edu \\
  \url{https://cs-people.bu.edu/sofya/}
\end{tocinfo}
\begin{tocinfo}[varma]
  Nithin Varma\\
  Assistant Professor\\
  Chennai Mathematical Institute\\
  Chennai, Tamil Nadu, India\\
  nithvarma\tocat{}gmail\tocdot{}com \\
  \url{https://www.cmi.ac.in/~nithinvarma/}
\end{tocinfo}
\end{tocauthors}

\begin{tocaboutauthors}
\begin{tocabout}[kalemaj]  %
  \textsc{Iden Kalemaj} is a \phd\ student in \href{https://www.bu.edu/cs/}{Computer Science at Boston University}, working under the wonderful supervision of \href{https://cs-people.bu.edu/sofya/}{Sofya Raskhodnikova}, in the fields of Sublinear Algorithms and Differential Privacy. 

\noindent
  Iden grew up in Vlor{\"e}, Albania.
  Her math teacher and school master, Dafina Brokaj, with a rare sense of humour and wisdom, instilled in Iden the confidence to pursue mathematics.

\noindent
  Iden completed her undergraduate studies in \href{https://www.math.princeton.edu/}{Mathematics at Princeton University}, and started to explore her interest in theoretical computer science via classes in graph theory and complexity theory. In her free time she enjoys climbing in the many Boston gyms and taking dance classes.
\end{tocabout}

\begin{tocabout}[raskhodnikova]   %
\textsc{Sofya Raskhodnikova} is a professor of \href{https://www.bu.edu/cs/}{Computer Science at Boston University}. Her office is located in the newly constructed \href{https://www.bu.edu/cds-faculty/explore/bu-center-for-computing-data-sciences/}{Center for Computing and Data Sciences}, colloquially known as the Jenga building. She received all her degrees (S.\,B., S.\,M., and \phd) from \href{https://web.mit.edu/}{MIT}, and then was a postdoctoral fellow at the Hebrew University of Jerusalem and the Weizmann Institute of Science. She was a Professor of Computer Science and Engineering at Penn State and held visiting positions at the Institute for Pure and Applied Mathematics at UCLA,
Boston University, Harvard University, and at the Simons Institute for the Theory of Computing at Berkeley. Sofya works in the areas of randomized and approximation
algorithms. Her main interest is the design and analysis of sublinear-time algorithms for combinatorial problems. She has also made important contributions to data privacy. Sofya has been a faculty member at \href{https://sigmacamp.org/}{Sigma Camp}, a residential summer science camp for kids, for several years. 
As far as her hobbies go, recall that she works on privacy.
\end{tocabout}

\begin{tocabout}[varma]
  \textsc{Nithin Varma} is a researcher working in the fields of Sublinear Algorithms and Approximation Algorithms. He got his \phd\ from the \href{https://www.bu.edu/cs/}{Computer Science Department at Boston University}, where he was fortunate to be advised by
\href{https://cs-people.bu.edu/sofya/}{Sofya Raskhodnikova}.
    After that, he was a postdoctoral researcher at the \href{https://cs.hevra.haifa.ac.il/?lang=en}{Department of Computer Science, University of Haifa}, where he was hosted and mentored by 
\href{http://cs.haifa.ac.il/~ilan/}{Ilan Newman} %
and
\href{https://sites.google.com/view/nogazewi}{Noga Ron-Zewi}. 
He was introduced to Theoretical Computer Science during his undergraduate years at the \href{https://minerva.nitc.ac.in/}{Department of Computer Science and Engineering at the National Institute of Technology Calicut} through the wonderful courses offered by
\href{https://people.cse.nitc.ac.in/muralikrishnan/}{Prof.\ K Muralikrishnan}. 
In his free time, Nithin enjoys meditation, singing Indian classical music and practicing the Mohiniyattam dance.
\end{tocabout}
\end{tocaboutauthors}

\end{document}